\documentclass[sigconf]{acmart}
\usepackage{epsfig,endnotes}

\usepackage{xcolor}

\usepackage{tabularx}
\newcolumntype{Y}{>{\centering\arraybackslash}X}

\usepackage{amsmath, amsthm, amsfonts}

\usepackage{amssymb}
\usepackage{array}
\usepackage{graphicx}
\usepackage{mdframed} % for 'pointer'
\usepackage{hyperref}
\usepackage{subfiles}
\usepackage{multirow}
\usepackage{cryptocode}
\usepackage{algorithm}
\usepackage{algpseudocode}
\usepackage{makecell}
\usepackage{refcount}
\usepackage{enumitem}
\usepackage[symbol]{footmisc}

\newtheorem{game}{Game}
\algnewcommand{\LeftComment}[1]{\Statex \(\triangleright\) #1}
\newtheorem{theorem}{Theorem}
\newtheorem{lemma}{Lemma}
\newtheorem{definition}{Definition}

\newtheorem{remark}{Remark}
\newtheorem*{theorem*}{Theorem}
\newtheorem*{lemma*}{Lemma}

\newcommand{\IE}{\mathbb{E}}
\newcommand{\IP}{\mathbb{P}}
\usepackage{subcaption}

\copyrightyear{2026}
\acmYear{2026}
\setcopyright{cc}
\setcctype{by}
\acmConference[CCS '26]{Proceedings of the 2026 ACM SIGSAC Conference on Computer and Communications Security}{November 15--19, 2026}{The Hague, Netherlands}
\acmBooktitle{Proceedings of the 2026 ACM SIGSAC Conference on Computer and Communications Security (CCS '26), November 15--19, 2026, The Hague, Netherlands}
\acmDOI{10.1145/3830454.3846801}
\acmISBN{979-8-4007-2871-6/2026/11}

\newif\iffullversion
\fullversiontrue   % arXiv build
\newcommand{\appref}[1]{%
  \iffullversion ~\ref{#1}\else ~\ref{#1} of~\cite{fullversion}\fi}
\begin{document}

\author{De Zhang Lee}
\affiliation{%
  \institution{National University of Singapore}
  \country{Singapore}
}
\email{dezhanglee@comp.nus.edu.sg}
\author{Han Fang}
\affiliation{%
  \institution{University of Science and \\ Technology of China}
  \country{China}
}
\email{fanghan@ustc.edu.cn}
\author{Ee-Chien Chang}
\affiliation{%
  \institution{National University of Singapore}
  \country{Singapore}
}
\email{changec@comp.nus.edu.sg}




%make title bold and 14 pt font (Latex default is non-bold, 16 pt)
% \title{Proof-of-Authorship for Diffusion-based AI Generated Content}
\title{Proof-of-Authorship for Diffusion-based AI Generated Content}

% \author{
% {\rm Anonymous Author}
% }

\begin{abstract}

Recent advancements in AI-generated content (AIGC) have introduced new challenges in intellectual property protection and the authentication of generated objects. We focus on scenarios in which an author seeks to assert authorship of an object generated using latent diffusion models (LDMs), in the presence of adversaries who attempt to falsely claim authorship of objects they did not create.
While proof-of-ownership has been studied in the context of multimedia content through techniques such as time-stamping and watermarking, these approaches face notable limitations. In contrast to traditional content creation sources (e.g., cameras), the LDM generation process offers greater control to the author. Specifically, the random seed used during generation can be deliberately chosen. By binding the seed to the author's identity using cryptographic functions, the author can assert to be the creator of the object. We refer to this stronger guarantee as proof-of-authorship, since only the creator of the object can legitimately claim the object.
This contrasts with proof-of-ownership via time-stamping or watermarking, where any entity could potentially claim ownership of an object by being the first to timestamp or embed the watermark.
We propose a proof-of-authorship framework involving  a probabilistic adjudicator who quantifies the probability that a claim is false. Furthermore, unlike prior approaches, the proposed framework does not involve any secret. We explore various attack scenarios and analyze design choices using Stable Diffusion 2.1 and XL as representative case studies.
\end{abstract}

\maketitle

\section{Introduction}

Latent diffusion models (LDMs) 
such as Stable Diffusion, DALL-E, and Midjourney,
have gained prominence for their ability to generate 
remarkably high-quality images from textual descriptions.
Generating such images require significant prompt engineering and 
tuning effort on the part of the author, and it is in the 
author's interest to protect their generated content against misattribution.  

There are several well-known proof-of-ownership approaches that enable  an author to claim ownership of digital objects.  One widely used approach  is time-stamping, where an owner claims ownership of an object by publicly time-stamping the object with a digital signature. However, this method alone is insufficient, as it is possible for anyone to timestamp any object without being its author, but by being the first to claim authorship publicly.  

Another approach embeds a secret digital watermark into the generated object using existing image watermarking schemes, after which the author proves ownership by demonstrating knowledge of the secret watermark. This post-diffusion approach, however, requires additional mechanisms to counter ambiguity attacks \cite{DBLP:conf/ih/LiC04, DBLP:journals/jsac/CraverMYY98}, in which an adversary either embeds its own watermark or constructs a forged watermark to dispute authorship. More importantly, despite advances over the past decades, defending against watermark removal remains challenging.

% In digital watermarking, an author proves ownership by embedding a secret into the generated object. Post-diffusion methods embed this secret in the pixel domain but are highly vulnerable to removal attacks and ambiguity attacks, where an adversary constructs a forged watermark to dispute authorship \cite{DBLP:journals/jsac/CraverMYY98}. Although cryptographic hashing can be used to mitigate ambiguity attacks \cite{DBLP:conf/ih/LiC04}, such watermarks are still susceptible to removals. 

Instead of embedding the watermark into the generated object, it can be embedded through pre-diffusion watermarking, where objects are generated using randomness derived from a secret \cite{wen2023treeringwatermarksfingerprintsdiffusion, yang2024gaussianshadingprovableperformancelossless, DBLP:journals/corr/abs-2410-07369}. For example, the starting point is sampled from a distribution parameterized by a secret key, as in Gaussian Shading~\cite{yang2024gaussianshadingprovableperformancelossless} or the Pseudorandom Error Correcting Code (PRC) watermark~\cite{DBLP:journals/corr/abs-2410-07369}. Pre-diffusion watermarking has demonstrated greater robustness than post-diffusion watermarking. Nonetheless, as with post-diffusion watermarking, an adversary may perform an ambiguity attack by manipulating the starting point to embed a competing watermark. Furthermore, both pre- and post-diffusion watermarking rely on a secret to prove ownership. Although zero-knowledge proofs can protect such secrets \cite{goren2025noiseprintsdistortionfreewatermarksauthorship}, verification remains computationally expensive, and a leaked secret makes subsequent claims ambiguous.

% The secret watermark can also be embedded through pre-diffusion based watermarking \cite{wen2023treeringwatermarksfingerprintsdiffusion, yang2024gaussianshadingprovableperformancelossless, DBLP:journals/corr/abs-2410-07369}, where  objects are generated with randomness derived from some secret. For example, the starting point is sampled from a distribution parameterized by a secret key,  as in Gaussian Shading.  Similar to pixel-based watermarking, without additional protection, an adversary might able to carry out ambiguity attacks.   Furthermore, both pixel and latent-based watermarks relies on an author's secret  to prove ownership. Although  zero-knowledge proofs can be incorporated to protect such secrets \cite{}, verification  is computationally expensive, especially for the high-dimensional objects.

Our approach shares similarities with current pre-diffusion watermarking techniques, wherein the initial starting point is sampled from a specific distribution. However, our framework differs in how this starting point is derived. Rather than relying on a secret key to parameterize the distribution, we cryptographically bind the author's identity and generation meta-data directly to the seed used for the LDM's random starting point.  This simple lightweight mechanism, completely free of any secrets, is already sufficient for proof-of-ownership. In fact, the mechanism  achieves beyond proof-of-ownership in the sense that only objects generated by LDM by the author can be claimed, in contrast to proof-of-ownership whereby existing known objects could be claimed. Consequently, we refer to this stronger guarantee as a proof-of-authorship (POA).

Under this framework, the random starting point of the LDM is cryptographically bound to the author's identity and some generation metadata. Specifically, the seed used to sample the random starting point is obtained by passing authorship information, such as the author’s identity, prompt,  and model version,  into a public keyed hash function, where the key is derived using the author's identity.  To claim authorship of an object $I$, the author demonstrates that $I$ is similar to an object generated using a starting point bound to the author, and with high statistical confidence,  the probability  $q$ of a randomly generated object which     happens to resemble $I$ is extremely small. Furthermore, via cryptographic binding, no probabilistic polynomial-time adversary can outperform random guessing, under the same query budget,  when forging a valid authorship claim. 
Moreover, the framework permits authors to append arbitrary randomness to the hash function's input. Altering this random value produces distinct cryptographic seeds, which subsequently map to different independent starting points for a single prompt, thereby maintaining the model's natural output diversity.

While the framework is conceptually simple, it involves two subtle challenges. First, because LDM generation does not constitute a cryptographic one-way function, certain inversion attacks may be feasible. Second, authorship claims rely on statistical hypotheses over rare events, making high-confidence validation with small sample sizes non-trivial.

To address this, we rely on two mechanisms (Section~\ref{sec:assumption}). First,  well-accepted cryptographic hardness assumptions (Assumption A1) limits any adversary to random guessing. Second, estimating the  forgery probability  $q$ requires at least $1/q$ samples, which is infeasible for small $q$. If the similarity score is instead known to follow a parametric distribution, this reduces to estimating the distribution's parameters, which requires far fewer samples. We therefore exploit  a distributional condition (Assumption A2): the similarity score between the contested object and randomly generated objects  admits a sub-exponential distribution. Moreover, there is strong evidence that A2 arises naturally from the structure of the LDM itself. We show that it is sufficient that the individual latent coordinates are sub-Gaussian (which we refer to as Assumption (A2-I)). There is likewise strong evidence that A2-I holds, as it is from the design goal of the LDM, and we verify it empirically.

% This acts as a vital safeguard against LDM-level manipulation. It is a weak assumption that holds in practice, as LDMs generate objects following a multivariate Gaussian distribution, which naturally yields sub-exponential cosine similarities. Consequently, A2 ensures foundational security against model tampering while enabling efficient statistical testing. Furthermore, even if A2 is deliberately violated (e.g., via a backdoored LDM), an adversary still cannot succeed as long as the verifier possesses comparable computational resources.

We introduce a Probabilistic Adjudicator (PA) to conduct statistical tests and a Judge to determine whether the claimed parameters are reasonable. Although the PA is probabilistic in the proposed framework, it can be derandomized by deriving its randomness deterministically from its inputs. Adjudication then becomes a deterministic function of publicly available parameters, making the procedure publicly executable and verifiable. Consequently, it could potentially be codified as a smart contract for deployment on public ledgers, which presents an interesting direction for future work.

% However, this division of responsibilities may vary across application scenarios.

\begin{figure}
    \centering
    \begin{subfigure}[t]{0.3\linewidth}
        \centering
        \includegraphics[width=\linewidth]{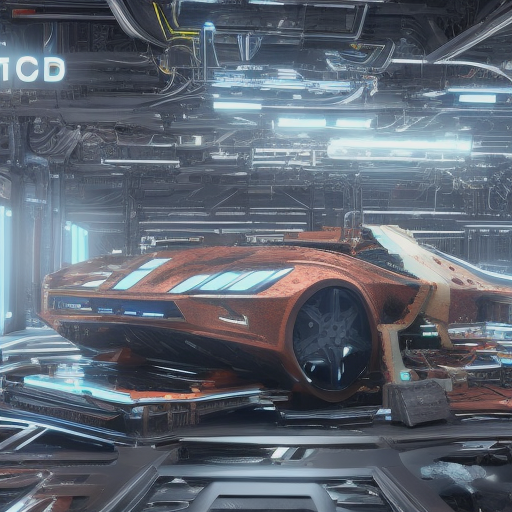}
        \caption{}
        \label{fig:intro_distort_orig}
    \end{subfigure}
    % \hfill
    % \begin{subfigure}[t]{0.32\linewidth}
    %     \centering
    %     \includegraphics[width=\linewidth]{fig/intro_jpeg_50.png}
    %     \caption{JPEG (quality 50)}
    %     \label{fig:intro_distort_jpeg50}
    % \end{subfigure}
    \hfill
    \begin{subfigure}[t]{0.3\linewidth}
        \centering
        \includegraphics[width=\linewidth]{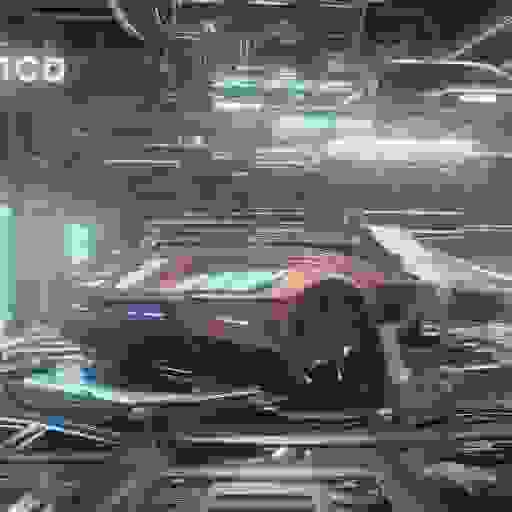}
        \caption{}
        \label{fig:intro_distort_jpeg1}
    \end{subfigure}
\hfill
    \begin{subfigure}[t]{0.3\linewidth}
        \centering
        \includegraphics[width=\linewidth]{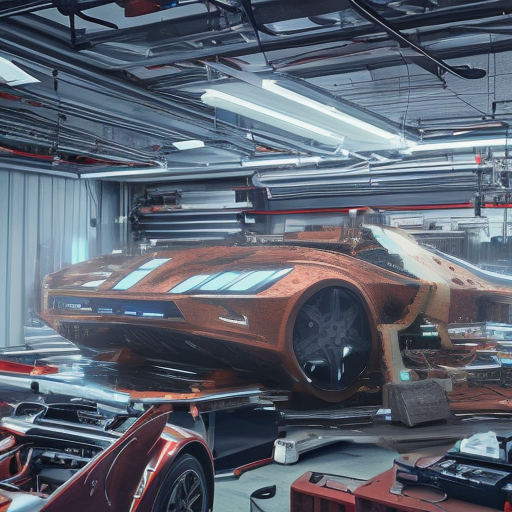}
        \caption{}
        \label{fig:intro_car_img_bg}
    \end{subfigure}
\caption{We observe similarity scores of $T = 8.6$ and $T = 7.5$ when
comparing the original image (a) to a JPEG-compressed version (b) and
an image with large overlap (c), respectively. Under the assumptions
of Section~\ref{sec:assumption}, the certified probability that a
random-forger generates an image achieving a score of $T$ or higher
is at most $2^{-50}$. Together with the cryptographic
guarantee provided by Theorem~\ref{thm:poa_PPT}, the object is attributed to the
author with high confidence.}
 
    \label{fig:intro_distort}
\end{figure}

In summary, our contributions are as follows:
\begin{enumerate}
    \item We propose a lightweight proof-of-authorship (POA) framework
    for LDM-generated content that requires no secrets, and provides a
    stronger assurance than proof-of-ownership.

    \item We identify two fundamental assumptions underpinning the
    framework: A cryptographic assumption (A1) for security, and a
    distributional assumption on the LDM (A2) for efficient
    adjudication. Under A1, we prove that any PPT strategy for forging
    proof-of-authorship generation parameters reduces to random
    guessing (Theorem~\ref{thm:poa_PPT}). Under A2, the random
    forger's success probability can be certified with a
    polylogarithmic number of samples, and there is strong evidence which supports A2 (Section~\ref{section:A2-analysis}).

    \item Our security analysis further shows that the framework
    remains secure across a range of attack scenarios, including
    forgery, a colluding or malicious PA and a malicious model provider
    (Section~\ref{sec:securityanalysis}).

    \item Empirical results on Stable Diffusion 2.1 and XL demonstrate
    that strong assurances are attainable in practice. Specifically, under  Gaussian noise, JPEG compression, affine transformation and regeneration attacks, a significantly per-query small random-forger success
    probability of $2^{-50}$ for proof-of-authorship is possible. 
    (Section~\ref{section:empirical}), see Figure \ref{fig:intro_distort} for example.
\end{enumerate}

\section{Background  and Notations}
\label{section:notations}
\subsection{Latent-Diffusion model}
A LDM takes in meta-parameter $m$ (e.g., type of scheduler, number of timesteps, etc.), an embedding $e$ generated from a prompt $p$ (for e.g., in SD 2.1, 
$e$ a real vector of dimension $n_{\tt p} \times 77 \times 1,024$ where $n_{\tt p}$ is a integer dependent on  $p$) , and a starting point $s$ in the latent space (e.g., in SD 2.1, this is a vector in $\mathbb{R}^{4 \times 64 \times 64}$), and outputs the generated image $L_m(e,s)$.  

The authors are typically oblivious to the numeric values in the embedding $e$. Instead, an author would provide a text prompt  which is then converted to an embedding $e$ through a text encoder. For e.g.,  the CLIP Tokenizer converts a text to the real vector of the embedding $e$ of dimensions $n_{\tt p} \times 77 \times 1,024$, where $n_{\tt p}$ is the number of tokens in the prompt.  

The starting latent point is randomly generated following a distribution $\mathcal{D}_0$ suitable for the diffusion process, for e.g., a multivariate Gaussian distribution. This starting point can be generated using some pseudorandom number sampling $G:\{0,1\}^{n_{ \tt seed}}\rightarrow \mathcal{L}$ that maps a  binary sequence of length $n_{\tt seed}$ to a point in the latent space $\mathcal{L}$, such that the output distribution $\mathcal{D}_0$ approximates the desired  multivariate Gaussian distribution when the seeds are uniformly chosen.  Small amounts of noise is iteratively removed over several timesteps, conditional on the prompt and LDM parameters, resulting in a latent space representation of the image. Finally, this latent space representation is upscaled into an actual image (typically of a higher dimension) using a variational autoencoder (VAE). 

To create an image,  an author could experiment with different prompts, seeds and meta parameters to derive a desired generated object $L_m (e, s)$, where $e$ is the embedding derived from the prompt, $s$ is the starting latent point generated from the seed, and $m$ is the meta parameter.   To summarize, from the author's perspective, with $m$, $p$ and a {\em seed},  the generated image is 
$   L_m ( e, s )$,
where $e = \mbox{CLIP} (p)$ and  $s = G ( \mbox{\em seed} ).$
% \[      L_m ( e, s ),
% \mbox{where\ \ \ \ } e = \mbox{CLIP} (p) \mbox{\ and\ }  s = G ( \mbox{\em seed} ).\]

The effect by $e$ is significant  as it corresponds to the semantics of the generated content, whereas $m$ corresponds to the visual quality, and different starting points $s$ lead to different generated content with similar semantics.  As a function, the LDM is assumed to admit some ``smoothness'' properties, in the sense that a slight perturbation of $e$, $s$ and $m$ will lead to small perturbations of the image in the pixel domain (under the $\ell_2$-metric) \cite{liang2024unravelingsmoothnesspropertiesdiffusion, yang2024lipschitzsingularitiesdiffusionmodels, guo2023smoothdiffusioncraftingsmooth}.  

\subsubsection{VAE Autoencoder} \label{subsubsection:vae}
In a typical LDM pipeline, the denoised latent is upscaled into an image of higher dimensions using a pretrained VAE. This is for computational efficiency. The VAE encoder maps an image $x$ into a latent $z$, which admits the conditional distribution
$q_\phi(z \mid x) = \mathcal{N}(z \mid \mu_x, \sigma_x^2 I)$,
where both $\mu_x$ and $\sigma_x$ are vectors dependent on $x$.  Sampling from this distribution yields latent space representations of images, which decode back into images perceptually similar to the original image $x$. Furthermore, VAEs are designed such that the unconditional prior distribution of the latent space $z$ follows a $\mathcal{N}(0, I)$ distribution. 
% In particular, the individual components in the latent space are independent.  This structural property will be used in our subsequent analysis. 

\subsection{Image Similarity Function  } \label{subsection:image_similarity_measure}
Measuring similarity between images in the pixel domain is challenging. Possible similarity functions include the $\ell_2$ metric or some metric derived from human visual system (see \cite{DU2020115713} for a survey). However, such metrics are sensitive to certain transformations (e.g. mild geometric transformations or small amounts of noise) which maintain perceptual similarity but induces large distances when these metrics are employed, and may be susceptible to higher false-positive rates.

To capture similarity under reasonable transformations which maintain perceptual similarity, we measure distances in the latent space rather than directly in pixel space. The latent representation provides a lower-dimensional embedding of the image and, and inherent robustness to small perturbations and noise due to the VAE training process \cite{Kingma_2019}.
This has been validated empirically, see for example \cite{valleti2023physicschemistryparsimoniousrepresentations, lee2025enhancingvariationalautoencoderssmooth}. 

% Consider a class of admissible perceptual-preserving transformations $\mathcal{T}$ (e.g., compression, blurring, rotation) and a distance function $d(\cdot, \cdot)$ on the latent space. Given a VAE encoder $\mathrm{enc}$, we define the similarity between two images $I_0$ and $I_1$ as the minimum latent space distance under transformations of $I_1$:
% \begin{equation} \label{eqn:similarity}
% \mathrm{Sim}(I_0, I_1) = \min_{T \in \mathcal{T}} d\bigl(\mathrm{enc}(I_0), \mathrm{enc}(T^{-1}(I_1))\bigr).
% \end{equation}
% In practice, identifying the most suitable $T \in \mathcal{T}$ that satisfies (\ref{eqn:similarity}) can be done either manually (e.g., through visual inspection) or algorithmically; see \cite{chen2024surveydeeplearningmedical} for a comprehensive survey. Similar considerations arise in other authentication schemes (e.g., digital signatures), where even a small change to a binary message $m$ (such as adding a padding or small bit shift) can yield a completely different hash. In such cases, an appropriate transformation is also needed, though in a different domain. In this work, however, the exact method of selecting the optimal transformation is not our main focus. Instead, we concentrate on the authentication aspect of the proposed proof-of-authorship scheme and assume that $T(I_1) = I_i$ in our subsequent analysis. Nonetheless. our empirical evaluations in Section \ref{section:similarity_eval} considers some transformations.

\subsection{Pseudorandom function family (PRF)} \label{subsection:prf}

A pseudorandom function family (PRF) is a class of functions whose outputs cannot be distinguished from a random oracle in polynomial time. A random oracle, on query that it has not seen before,  randomly outputs a value sampled from a desired distribution, otherwise (on a previously seen query), outputs the same previously sampled value. Many practical constructions, e.g., HMAC with SHA, are widely  assumed to be PRF.   

Let us denote the class of functions $f_i : \{0,1\}^{n_{\mathsf{k}}} \to
\{0,1\}^{n_{\mathsf{seed}}}$ to be the PRF, where $i \in \mathcal{I}$ for a
set of identities $\mathcal{I}$, $n_{\mathsf{k}}$ is the bit length of the
input and $n_{\mathsf{seed}}$ is the bit length of the seed. Let
$\mathcal{F} = \{f_i : i \in \mathcal{I}\}$ denote the function family. If the key $i$
is public, standard PRF security does not apply, and we instead require
$\mathcal{F}$ to satisfy two properties stated as assumptions on the family:
collision resistance across identities and correlation intractability,
explained below. These properties have been used extensively in many
cryptographic protocols to obtain random oracle-like  guarantees in the
standard model, a notable example being the Fiat--Shamir
transformation~\cite{shamir2000prove}.

In the context of LDMs, one could employ $\mathcal{F}$ to generate the seeds
that are fed into the deterministic  pseudorandom sampler $G$ to obtain the
starting points from the starting distribution $\mathcal{D}_0$. 
{ %%\color{blue}
\paragraph{Collision Resistance}
Standard PRF security is defined with respect to a function sampled uniformly
from the family under a \emph{secret} key: indistinguishability from a random
oracle holds only against adversaries who do not know the key, and no
guarantee is provided once the key is known. When keys are public or
adversarially chosen, two stronger properties are required. First, that
$f$ remains unpredictable on any previously unevaluated key--input pair
even under adversarially chosen or known keys; this is formalized as
correlation intractability below. Second, that no PPT adversary can find,
except with negligible probability, $(i, x)$ and $(j, x')$ with $i \neq j$
and $f_i(x) = f_j(x')$. This notion of collision resistance is also known as
\emph{full robustness} in~\cite{cr_prf_2017}.

\paragraph{Correlation Intractability}
Let $R$ be a relation whose entries take the form $((i, x), s)$. We say
$R$ has \emph{density} $\rho$ if for every $(i,x)$,
$$\Pr_{s \leftarrow \mathcal{S}}[((i,x), s) \in R] \le \rho,$$ where $s$ is
sampled uniformly from $\mathcal{S}$; $R$ is \emph{sparse} if $\rho$ is
negligible (w.r.t. the security parameter). A function family $\{f_i\}$ is correlation
intractable (CI) for $R$~\cite{canetti2004random} if every PPT adversary
making at most $Q$ evaluations of the family outputs $(i, x)$ with
$((i,x), f_i(x)) \in R$ with probability at most
$Q \cdot \rho + \mathrm{negl}(\lambda)$. Intuitively, this means that no adversary is able to search for an input $x$ s.t. $f_i(x)$ is in a predetermined set better than random guessing. No fixed family is CI for all
sparse relations~\cite{canetti2004random}, and known standard-model
constructions~\cite{canetti2019fiatshamir, peikert2019nizk} cover only
restricted relation classes; we therefore assume CI only for the specific
relation family of Appendix~\appref{appendix:poa_PPT}, an assumption with
the same evidentiary status as modeling $\{f_i\}$ as a random oracle. The
same style of assumption underlies standard-model soundness of the
Fiat--Shamir transform~\cite{shamir2000prove,
canetti2019fiatshamir}.

}

\subsection{Statistical Testing} \label{subsection:statistical_testing}
% We review the statistical estimation problem underlying our
% verification procedure, which involves certifying that an observed value is
% exceptionally rare under a reference distribution.

Let $W$ be a real-valued random variable with distribution
$\mathcal{W}$, from which independent samples can be drawn, and let
$T$ be an observed value. The quantity of interest is the tail
probability
\[
q \;=\; \Pr[W \ge T],
\]
i.e., the probability that a fresh draw from $\mathcal{W}$ attains a
value at least as large as the one observed. We are interested in estimating $q$. 

\paragraph{Estimation of $q$}
To ascertain whether $q \le q_0$ empirically, one can draw independent samples from $\mathcal{W}$ and compute the fraction of samples satisfying the event $\{W \ge T\}$. However, obtaining high confidence in this test requires at least $1/q_0$ samples, which is intractable for small $q_0$ (e.g., $q_0 = 2^{-50}$). If instead $\mathcal{W}$ is assumed to belong to a
known parametric family, the problem reduces to estimating the
distribution parameters, which requires far fewer samples. 
% Nonetheless, the cost is estimation error, which may lead to under-estimation of $q$.
Appendix~\appref{appendix:p_estimate} shows that under the
distributional assumptions adopted in this work, approximating $q$
with additive error at most $\alpha$ and confidence $1-\delta$
requires $n = \Omega\big(\ln^2 \tfrac{1}{\alpha} \ln
\tfrac{1}{\delta}\big)$ samples. Henceforth, $\hat{q}$ denotes the
point estimate of $q$ computed from $n$ independent samples, and
$\mathcal{C}_{\delta} = (\hat{q}-\alpha,\, \hat{q}+\alpha)$ the
corresponding $(1-\delta)$ confidence interval of width $2\alpha$.

\subsection{Summary of notations}
Table \ref{table:notation} presents a summary of notations used, including those introduced later. 
\begin{table}[t]
\caption{Summary of notations used.}
\label{table:notation}
% \small
\begin{tabularx}{\columnwidth}{@{}lX@{}}
\toprule
Symbol & Description \\
\midrule
$m$ & Meta-parameters for LDM \\
$e$ & Embedding \\
$r$ & Free random bits in PRF input \\
$\kappa$ & Generation parameters $\langle m,e,r\rangle$ \\
$i \in \mathcal{I}$ & Author identity from set of all possible identities $\mathcal{I}$\\
$s$ & Starting point of LDM \\
$L_m(e,s)$ & LDM generated object \\
$u \in \mathcal{U}$ & Seed, $\mathcal{U}=\{0,1\}^{n_{\mathsf{seed}}}$ \\
$G(\cdot)$ & Deterministic sampler, seed to starting point \\
$\mathcal{D}_0$ & Distribution of $G(u)$ for uniform $u$ \\
$\mathcal{F}$ & Keyed PRF family $\{f_i : i \in \mathcal{I}\}$ \\
$\lambda$ & Security parameter of $\mathcal{F}$ \\
$Q$ & Adversary's query budget on $\mathcal{F}$ \\
$p_r$ & Target random-forger success rate \\
$R_{I,i}$ & Forgery relation (Appendix~E) \\
$\mathrm{Sim}(\cdot,\cdot)$ & Similarity function \\
$T$ & Observed similarity score \\
$q$ & Probability a random-forger attains score $\ge T$ \\
$\hat{q}$ & Point estimate of $q$ \\
$\mathcal{C}_\delta$ & $(1-\delta)$ confidence interval of $q$, width $2\alpha$ \\
$\alpha$ & Bound on estimation error of $q$ \\
\bottomrule
\end{tabularx}
\end{table}
% \begin{table}[!htbp]
% \caption{Summary of notations used.}
% \label{table:notation}
% \centering
% \begin{tabular}{l p{.8\columnwidth}}
% \hline
% Symbol & Description \\
% \hline
% $m$ & Meta-parameters for LDM. \\
% $e$ & Embedding for LDM generation. \\
% $s$ & Starting point of LDM. \\
% $L_m(e,s)$ & LDM generated object from starting point $s$ and embedding $e$. \\
% % $\mathcal{L}$ & Latent space. \\
% $\mathcal{D}_0$ & Starting point distribution that approximates multivariate Gaussian distribution. \\
% $G(\cdot)$ & A deterministic sampler that gives $\mathcal{D}_0$ from uniform input. \\
% $f_i(\cdot)$ & Pseudorandom function for identity $i$. \\
% $\mathrm{Sim}(\cdot,\cdot)$ & Similarity function. \\
% % $H_0, H_1$ & Null and alternative hypothesis. \\

% $T$ & Observed similarity score. \\
% $q$ & Probability that a random-forger attains similarity score at least $T$. \\
% $\hat{q}$ & Point estimate of $q$. \\
% $\mathcal{C}_{\delta} $ & $(1-\delta)$ confidence interval of $q$ with width $2\alpha$.\\
% $\alpha$ & Upper bound of estimation error of $q$ with probability of at least $(1-\delta)$.\\
% \hline
% \end{tabular}

% \end{table}

\section{Proof-of-Authorship Framework}

 This framework involves several  entities: {\em Author}, {\em Probabilistic adjudicator}, {\em Judge} and {\em forger};  three phases: {\em setup}, {\em generation} and {\em contention}; and several algorithms: the LDM $L_m( \cdot, \cdot)$,  similarity function ${Sim}(\cdot, \cdot)$,  the class of  cryptographically secure pseudorandom function $f_i$'s and pseudorandom number sampling $G(\cdot)$ as discussed in Section \ref{section:notations}, and the statistical test to be carried out by the PA. 
 % Note that all the algorithms are deterministic although they might take in inputs that are randomly sampled by an entity.

\subsection{Phases} \label{subsection:phases}
\subsubsection{Registration}\ \ \  
During this phase, each author chooses and publishes a unique identifier $i\in \mathcal{I}$ that serves as the index $i$ of $f_i(\cdot)$. 
The purpose of registration is to avoid the ambiguous situation where multiple authors have the same identifier. 

\begin{figure*}
    \centering
    \includegraphics[width=1\linewidth]{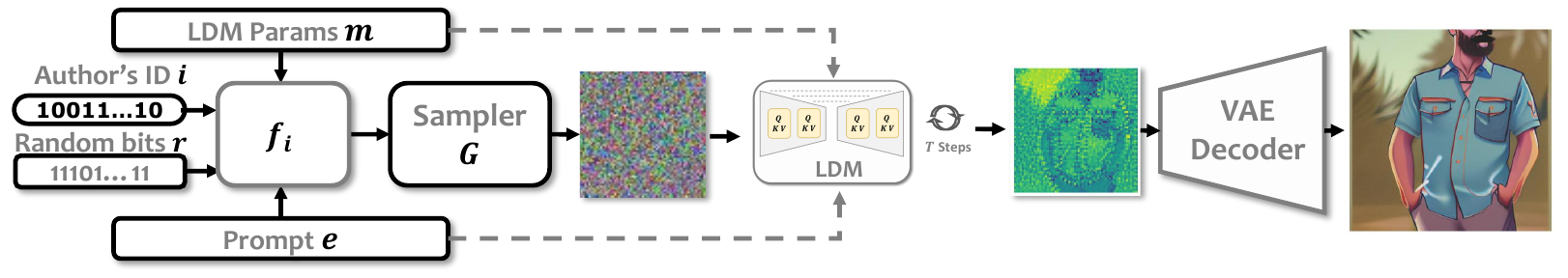}
    \caption{An illustration of the generation phase for the POA $\langle m, e, r \rangle$. A pseudorandom function $f_i$ (e.g. HMAC-SHA3) deterministically maps the POA to a seed that initializes the starting point. The random free bits $r$ enable sampling from different locations, preserving diversity in the generated outputs.}
    \label{fig:poa_generation}
\end{figure*}

\begin{figure*}
    \centering
    \includegraphics[width=1\linewidth]{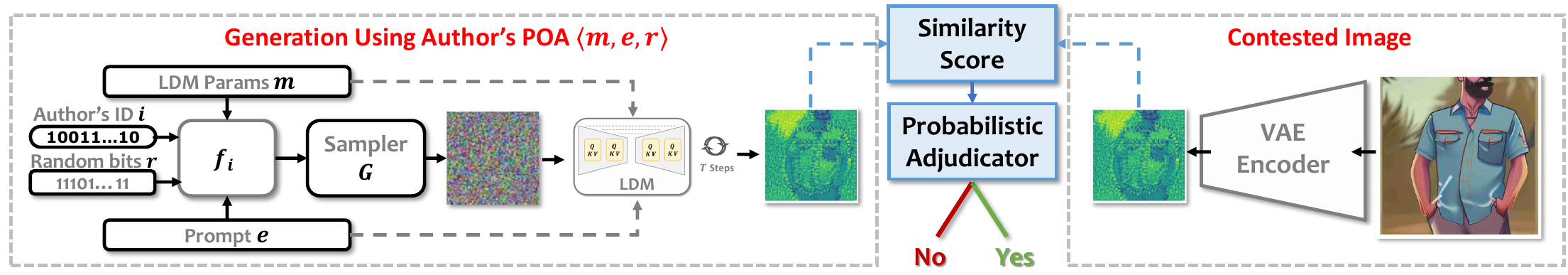}
    \caption{An illustration of the contention phase in our POA framework. The left depicts the VAE latent generated from the POA $\langle m, e, r \rangle$, compared against the VAE latent of the contested image (right). A similarity score is computed from the two latents (Section~\ref{subsection:similarity}) and evaluated by the probabilistic adjudicator to determine whether the POA corresponds to the contested image. }
    \label{fig:poa_contest}
\end{figure*}
\subsubsection{Generation}  \label{subsubsection:generation} \ \ \
An author can experiment with various prompts (which determines the embedding), meta-parameters and starting points. In this framework, the starting points are deterministically obtained from the composition
$ (G\circ f_i) ( \langle m, e, r \rangle )$, 
where $f_i$ is a function from the function family $\mathcal{F}$, $i$ is the author's identity, $m$ is the meta-parameter, $e$ the embedding, and $r$  a randomly sampled binary sequence.   Hence, if $m$ and $e$ are fixed, the author $i$ can randomly sample the starting points by feeding in a random sequence $r$ to generate multiple objects. An illustration of this generation process is given in Figure \ref{fig:poa_generation}.
The parameters $m, e, r$ of an object $J$  would be archived for further lookup. As it is not necessary to keep $m, e, r$ as secret, if the author chooses to, they can be pushed to public domains or distributed ledgers. Henceforth, we will use  $\kappa=\langle m, e ,r \rangle$ to denote the generation parameters of an object.

\subsubsection{Contention}\ \ \ \label{subsubsection:PA}
This situation arises when an author $i$ wants to claim authorship of an object $I$.   For instance,  the object $I$ is an  image that appears in the public domain and the author wants to convince others that $I$ is a slightly distorted and transformed copy of an original image $J$. In particular, the author knows $m, e$ and $r$ such that $J= L_m ( e,  (G \circ f_i)  (\langle m, e, r\rangle  ))$.  Let us call $J$ the {\em original} object and call $I$ the {\em contested} object. 
To support the claims, the following is carried out:

\begin{enumerate}
\item The author passes the following to  the { probabilistic adjudicator}: contested object $I$,  the author's identity $i$,  the generation parameters $\kappa =\langle m, e,r\rangle$, a desired error bound $\alpha$, statistical confidence $(1-\delta)$, and a transformation $t$ which aligns the contested and original objects.
 The PA carries out the statistical test outputs $\hat{q}$. The parameter $\alpha$ and $q$ represent the range $\mathcal{C}_\delta = (\hat{q}-\alpha, \hat{q}+\alpha)$, which is the $(1-\delta)$ confidence interval of $q$. 
% where the randomness is over the starting point $s$ sampled from a distribution $\mathcal{D}_0$, and $p$ satisfies
% $
% p = \IP[Sim  ( t(L_{{m}} ( {e}, s )), I ) \geq T]
% $
% If the PA unable to find a $W< xx$, then the PA output ``un-attainable''.
 An illustration of the contention phase is given in Figure \ref{fig:poa_contest}.

\item The PA passes all the information above to the judge, including its inputs, $\hat{q}$ and $\mathcal{C}_\delta$
\end{enumerate}

\paragraph{Judge}
Using the information provided by the PA, the judge outputs an assertion
on whether the contested object $I$ corresponds to author $i$ and
generation parameters $\kappa$. The primary criterion is whether the
upper endpoint of $\mathcal{C}_\delta$ is sufficiently small (e.g.,
$2^{-50}$), so that it is exceedingly unlikely that a forger
have produced an object this similar by chance. The judge may also
weigh other information, e.g., the validity of the supplied
transformation or other contextual factors.

Section \ref{subsection:PA_spec} describes the PA and judge in greater detail. 

\subsubsection{Forger}
A forger intends to falsely claim authorship of a point of contention $I$.  The object $I$ could be generated from LDM or objects from other non-AIGC source (e.g. digital image taken using a camera). In cases where the object is LDM generated, let us assume that the forger knows all the parameters, i.e., the original author's identity $i$ and corresponding   generation parameters $\kappa=\langle m,e,r \rangle$ that generates $I$.

Given $I$, the forger wants to find a $\langle \widetilde{m}, \widetilde{e}, \widetilde{r}\rangle$ s.t. the PA outputs a small $\hat{q}$.
Since the forger knows $m$ and $e$, the forger could try multiple starting points with $m$ and $e$ to find an object similar to $I$.  Alternatively, the forger could try with different meta-information and
embedding.

% \subsubsection{Remarks} \label{subsubsection:remarks}
% \begin{enumerate}
% \item {\em Verifiable PA. \ \ }  Since the assertion by PA can be publicly verified, it does not need to be a separate entity.  Nevertheless, we prefer to take it as a separate entity to anticipate contentions when an author wants to consider other similarity function, for e.g. a new lossy compression or transformation that was not considered before, or a different metric.  In addition,  certain similarity measures are expensive to compute and it is more practical to consider some fast approximation. In such cases, an independent PA can provide expert opinion on the admissibility of various alternatives.

% \item {\em Witness of the estimate $p$. \ \ \ } The optional witness consists of another seed $s_0$ that generates a list of cryptographic secure pseudorandom starting point $s_1, \ldots s_{n'}$  and a list of real values $a_1, \ldots, a_{n'}$ for some $n'$,  where each $a_j$ is the PA's estimate of the ratio
% $$  a_j \approx \frac{\mbox{Sim}  ( L_{\widetilde{m}} ( \widetilde{e}, s_j ), I )}
% {\mbox{Sim} ( J, I)}.
% $$
% To  verify the PA's assertion, a public can choose any index and verify that the corresponding ratio is  accurately approximated.  Although not within scope of our security formulation,   the seed $s_0$ can prevent a scenario where a malicious PA only selects starting points that are favorable to an entity. 

% \end{enumerate}

% Both of these issues are addressed in detail in our concrete specification of the PA in Section~\ref{subsection:PA_spec}.

\subsection{ Assumptions}
\label{sec:assumption}
Our framework relies on two assumptions, an assumption on  cryptographic hardness, and  statistical properties of the LDM.

\begin{enumerate}
\item[(A1)]
{\em Hardness of $f_i$'s. }   We assume that $\{ f_i \}_{i\in \mathcal{I}}$ is from a cryptographically secure pseudorandom function (PRF) family, which satisfies the collision resistance and correlation intractability properties described in Section \ref{subsection:prf}.  

\item[(A2)]
{\em Diffusion property on LDM. }  We assume that, given any $i_0$, $m$, $e$, the distribution $Sim(i_0, L_m (e, S))$ where $S\sim \mathcal{D}_0$, 
follows a sub-exponential probability distribution function. %non-central Chi-square distribution. 

\end{enumerate}
The above assumptions are crucial. By Assumption A1, a forger cannot
perform better than the random-forger, who searches uniformly at random
for an alternative set of parameters corresponding to the image in question
(Section~\ref{sec:securityanalysis}). We instantiate $\mathcal{F}$ as
HMAC-SHA3, for which the properties in A1 are assumed as discussed in
Section~\ref{subsection:prf}.

{%\color{blue}
Assumption A2 is what makes adjudication tractable. Without a distributional characterization, bounding the tail probability of $Sim(i_0, L_m(e, S))$ at the $(1-q)$-th quantile with high statistical confidence requires at least $\frac{1}{q}$ samples, which is intractable for small $q$, e.g., $q = 2^{-50}$. Under Assumption A2, the same tail probability can instead be certified by estimating the parameters of the underlying sub-exponential distribution from only $n = \Omega\left(\ln^2\frac{1}{\alpha}\ln\frac{1}{\delta}\right)$ samples
(Theorem \ref{thm:p_min_samples}, Table \ref{table:sample_complexity_comparison}). Section \ref{section:A2-analysis} discusses the motivation and validity of Assumption A2 in greater detail.

}

\subsection{Choice of  $G \circ f_i$}  \label{subsection:gf}
We choose HMAC-SHA3 as our choice of function satisfying A1. For the sampler $G$, we adopt a pseudorandom number generator in the PyTorch library, that on an input seed, outputs pseudorandom samples following the $\mathcal{N}(0,I)$ distribution. Our security analysis requires the property that, if $s$ is computationally indistinguishable from a truly random $s'$, then $G(s)$ is computationally indistinguishable from $G(s')$. Since this property holds for any deterministic $G(\cdot)$, it is not necessary for $G$ to be a cryptographically secure pseudorandom number generator. 

% We choose an HMAC with SHA3-256 as $f_i$'s.  The  tuple $\langle m, e, r\rangle$  is represented using string concatenation of $m$, $e$ and $r$ with special symbols as deliminator.   For the Gaussian sampler $G$, we adopt functions in Python's PyTorch library, that maps a seed to the $\mathcal{N}(0, I)$ distribution. HMAC SHA3-256 is widely assumed to be a PRF, and we shall assume so and thus, Assumption (A1) holds for $f_i$. We do not impose any further assumptions on $G$: If $s$ is computationally indistinguishable from a truly random sequence $\hat{s}$, then $G(s)$ is also computationally indistinguishable from $G(\hat{s})$.

\subsection{Probabilistic Adjudicator And Judge} \label{subsection:PA_spec}

% As mentioned in Section \ref{subsubsection:PA}, the input to PA are:  Contested object $I$, transformation $t$, generation parameters $\kappa = \langle m,e,r\rangle$, author's ID $i$, upper bound of estimation error $\alpha \in (0,1)$,  and statistical confidence  $(1-\delta) \in (0,1)$.

\subsubsection{PA Algorithm}
The PA performs the following steps. 
% \paragraph{\textit{Step 1: Setup}} 

\paragraph{{Step 0: Determining The Number of Samples Needed}}
First, the PA
generates the original $J = L_m(e, G(f_i(\kappa)))$. Recall that $W$ is the distribution of $ L_m(e, s)$ where $s \sim \mathcal{D}_0$. 
To estimate $W$, the PA computes $n = \ln^{2} \frac{1}{\alpha} \ln \frac{1}{\delta}$, which is the number of Monte-carlo samples needed so that the estimation error of $\hat{q}$ is at most $\alpha$ with confidence of at least $(1-\delta)$. Section \ref{subsection:PA_spec} gives details on the derivation of $n$.

 \paragraph{{Step 1: Estimating the Similarity Score Distribution}}
 The PA 
 independently samples $s_1, \dots, s_{n} \sim \mathcal{D}_0$, and generates $J_1, \dots, J_{n}$, where $J_j = L_m(e, s_j)$. Denote the similarity function $Sim(X, Y) = \frac{1}{d} X \cdot Y$, where $d$ denotes the number of elements in $X$. Subsequently, the PA approximates the empirical distribution $\widehat{W}$ of the Monte-Carlo samples $\{Sim(J, t(J_j)):j=1, \dots, n\}$, by computing the maximum likelihood estimate of the underlying distribution parameters over these samples. We use the implementation provided by Python’s {scipy} package to perform this estimate.

% \textbf{Verifiable PA.} To ensure reproducibility and auditability, the computes a deterministic seed $H(\kappa)$, where $H(\cdot)$ is a pseudorandom generator (PRG). This seed will be used as the source of randomness. 

\paragraph{{Step 2: Computing $\mathcal{C}_{\delta}$}} The PA estimates $\hat{q} = \IP[\widehat{W} \geq Sim(J, t(I))]$. With probability of at least $1-\delta$, the error of $\hat{q}$ is at most $\alpha $. Accordingly,  $\mathcal{C}_{\delta} = \left( \hat{q} - \alpha, \hat{q} + \alpha\right)$ is a conservative $(1-\delta)$ confidence interval of $q$. 

% Next, the assertion threshold $W_\alpha^I$ is computed as defined in (\ref{eqn:W}), which corresponds to the $1-\alpha$ quantile of the similarity score distribution $W$. Since $\hat{W}$ is an empirical approximation of $W$, we estimate  $W_\alpha^I$ by constructing $\mathcal{C}_{\delta, W_\alpha^I}$, and take its upper quantile. This ensures that the resulting random attacker's success rate is at most $\alpha$ with high statistical confidence. 
% Without Assumption A2, a large $n$ is necessary to estimate $W_\alpha^I$ for small $\alpha$. In Appendix \ref{apppendix:estimate_quantile}, we give a method to estimate $\mathcal{C}_{\delta, W_\alpha^I}$, where the length of $\mathcal{C}_{\delta, W_\alpha^I}$ scales with $O\left(1/{\sqrt{n}}\right)$. 

\paragraph{{Step 3: Output}}
The PA passes all information above to the judge, including its inputs, $\hat{q}$ and $\mathcal{C}_\delta$.

\subsubsection{Judge}
The judge uses the information provided by the PA  to evaluate the assertion whether the contested object $I$ originates from the original $J$, generated using parameters which are bound to the author. Let $q'$ be the upper bound of $\mathcal{C}_\delta$. In our evaluations in Section \ref{section:empirical}, we simulate a judge by accepting the assertion if $q'$ is very small (e.g. $2^{-50})$, making it exceedingly unlikely that a random-forger could succeed.

\subsubsection{Remarks}%\hfill\par
\hfill\par
\noindent \emph{Role of $t$ in Step 1.}
A forger may attempt to evade detection by a transformation $t_0$ to the contested object. To prove authorship, the  author specifies the transformation $t$ which approximates $t_0$, and the similarity is to be computed on the transformed samples $t(J_j)$. The judge may reject the transformation if it is deemed unreasonable (e.g., if it corresponds to only a small subimage of the original).

% \hfill\par
\hfill\par
\noindent \emph{Publicly Executable and Verifiable PA.}
In Step 1, the samples are randomly chosen. Alternatively, we can impose that the randomness must be derived deterministically from the inputs. In this setting, the PA derives its randomness (i.e. the pseudorandom seed) from a secure deterministic one-way hash (e.g. SHA-256) of its inputs. Since the PA is deterministic, its outputs are now publicly verifiable. 
    
\hfill\par
% \hfill\par
\noindent \emph{Correctness and Optimality of PA.}
The PA described computes the test statistic for the hypothesis test formalized in Theorem~\ref{theorem:hypothesis}. For known scale and shape parameters, this test is
uniformly most powerful; the PA implements a plug-in version in which
these parameters are estimated from Monte Carlo samples, with the
estimation error accounted for through $\mathcal{C}_\delta$.

\hfill\par
\noindent \emph{Derivation of $n$ and $\mathcal{C}_{\delta}$.} 
In Step 0, Assumption A2 is used to reduce the number of Monte-Carlo
samples from $1/\alpha$ to $n = \ln^2 \frac{1}{\alpha} \ln \frac{1}{\delta}$.
Theorem~\ref{thm:p_min_samples}  shows that
$n = \Omega\left( \ln^2 \frac{1}{\alpha} \ln \frac{1}{\delta}\right)$ samples
suffice to bound the estimation error of $q$ by $\alpha$.
\newcounter{theorembackup}
% Save current theorem number
\setcounter{theorembackup}{\value{theorem}}
\setcounter{theorem}{\getrefnumber{thm:p_min_samples}-1}
\begin{theorem}
Under Assumption~A2, and standard MLE regularity assumptions,  let \( \alpha \in (0,1) \) and \( \delta \in (0,1) \) and consider a fixed threshold \(T\). Let
\( \hat{q} = \Pr[{\widehat{W}} \ge T] \),  where $\widehat{W}$ is estimated over $n$ independently and identically distributed samples drawn from a distribution $W$. If $
n \in \Omega\left(
\ln^{2}\frac{1}{\alpha}\ln\frac{1}{\delta}
\right)
$, then $\hat{q}$ estimates the true probability $q=\IP[W \geq T]$
within an additive error of \( \alpha \), with probability at least \( 1-\delta \).
\end{theorem}
\setcounter{theorem}{\value{theorembackup}}
The proof of Theorem \ref{thm:p_min_samples} is deferred to Appendix~\appref{appendix:p_estimate}.  Table~\ref{table:sample_complexity_comparison} reports the number of samples required when Assumption~A2 is employed, versus naive sampling.

\begin{table}[t]
\caption{Comparison of sample complexity for determining whether \(q \le \alpha\) with confidence \(1-\delta\), where \(\delta = 0.001\). Leveraging Assumption A2 reduces the required number of samples from exponential in \(1/\alpha\) to polylogarithmic.}
\label{table:sample_complexity_comparison}
\centering
\small
\begin{tabular}{|c|c|c|c|}
\hline
\textbf{} &
\textbf{(I) Naive} &
\textbf{(II) Assumption A2} &
\textbf{Ratio of} \\
\textbf{\(\alpha\)} &
\textbf{Sampling } &
\textbf{(\(n=\ln^{2}(1/\alpha)\ln(1/\delta)\))} &
\textbf{I:II} \\
\hline
\(2^{-10}\) & \(\geq 2^{10}\) & \(\approx 2^{8}\) & \(2^2 \) \\
\hline
\(2^{-30}\) & \( \geq 2^{30}\) & \(\approx 2^{12}\) & \(2^{18}\) \\
\hline
\(2^{-50}\) & \(\geq 2^{50}\) & \(\approx 2^{13}\) & \(2^{37}\) \\
\hline
\end{tabular}

\end{table}

% \begin{remark}[\emph{Optional Judge}]
% Instead of making a hard classification decision, the PA can forward its output to a judge. The judge may interpret the $p$-value as a quantitative measure of the likelihood that the two objects share the same POA, and incorporate it as part of the overall decision-making process (as illustrated in Figure~\ref{fig:poa_contest}). 
% % For example, if the $p$-value is slightly larger than the cutoff threshold $\alpha$, the judge may also consider the validity of the transformation applied by the PA, as well as additional contextual evidence (e.g. timestamps) submitted by the author.
% \end{remark}

\section{Security Analysis}
\label{sec:securityanalysis}

% Suppose an image is generated by an author with identity $i$ and generation parameters $\langle m, e, r\rangle$. The forger’s goal is to (1) generate an object can be attributed to the forger’s own identity and generation parameters, or (2) transform the image  into a modified version which is similar to the original but cannot be claimed by the author. we consider the following threat models:

We make the following assumptions about the system and trust model. Firstly, the
diffusion model is public. Secondly, the author's identifier $i$, generation
parameters $\kappa = \langle m,e,r\rangle$, and the corresponding
generated object $I$ are also public. Third, there exists a trusted identity
system to broadcast and validate authors' identities. Finally, we assume that the judge is
honest.

Under the above model, we consider the following  attack scenarios:
\begin{enumerate}
    \item \textit{Forger} (Section~\ref{subsubsection:forger_attack_strategy}):
    A forger with identity $j \neq i$ attempts to produce generation
    parameters $\widetilde{\kappa}$ that generate an object similar to $I$.
    \item \textit{Malicious Model Provider} (Section~\ref{subsubsection:collusion}):
  The diffusion model is developed or distributed by a malicious provider  seeking to assist forgers.
    \item \textit{Colluding PA} (Section~\ref{subsection:colluding_PA}):
   A malicious PA deliberately manipulates the sampling process to benefit
    the forger.
    \item \textit{Robustness} (Section~\ref{subsection:adv_attack}):
  An adversary perturbs an object while preserving its perceptual content,
    thereby preventing the legitimate author from proving authorship.

\end{enumerate}

The first three scenarios aim to make an invalid authorship claim appear
valid by biasing the estimated random-forger success probability $q$ in
favor of the forger. The final scenario, in contrast, aims to make a valid
claim appear invalid by mounting an adversarial attack against the
similarity function.

% - The model is publicly known.  
% - The paremeter <m,e,r,i> is also public. 
% - A trusted identity system.
% - The Judge  is ....

%described in Section \ref{sec:threatmodel}.

% In our security analysis, we show that no polynomial-time forger can achieve more than a negligible advantage for the first goal over an adversary limited to random guessing (i.e., a random-forger).  The second goal corresponds to a \textit{removal attack}, which is addressed in our subsequent analysis of the similarity function.

% \subsubsection{Threat Model} \label{sec:threatmodel}

% Given a publicly released set of generation parameters 
% $\kappa = \langle m, e, r \rangle$ associated with an author $i$ used to generate an object $I$, 

% Additionally, in Section~\ref{subsection:adv_attack}, we consider a {removal attack}, in which an adversary seeks to evade detection by manipulating the similarity function, so that the contested and original objects remain close under a given distance metric, but their similarity score is reduced. 

\subsection{Forger}   \label{subsubsection:forger_attack_strategy}

Given an image $I$ (generated by author $i$ and generation parameters $\kappa= \langle m, e, r\rangle$), a forger aims to produce generation parameters 
$\widetilde{\kappa} = \langle \widetilde{m}, \widetilde{e}, \widetilde{r} \rangle$, such that the resulting generated image
$
J = L_{\widetilde{m}}\bigl(\widetilde{e}, (G \circ f_j)(\langle \widetilde{m}, \widetilde{e}, \widetilde{r} \rangle)\bigr)
$
is similar to $I$ but corresponds to a forger selected identity $j$ and generation parameters. If this occurs, an \textit{ambiguous scenario} arises: both the original author and the forger possess valid but distinct 
identity and generation parameters pairs, namely $(i,\kappa)$ and
$(j, \widetilde{\kappa})$, respectively, which correspond to similar objects $I$ and $J$. 
% Let us call this forger a malicious forger. 

% Empirically, if an adaptive forger uses a different embedding $\widetilde{e}$ and meta-parameter $\widetilde{m}$ compared to the original $e$ and $m$, the forged object $J$, while possibly semantically similar, will likely be significantly different compared $I$ (see \cite{shen2024promptstealingattackstexttoimage} for the empirical results). 
% Since our framework does not require $\kappa$ to be kept secret, the forger is free to use $m$ and $e$. Nonetheless, our subsequent analysis does not restrict the forger to use the author's $(m,e)$, and they are free to derive another set of $(\widetilde{m}, \widetilde{e})$ in their forgery attempt. 

\subsubsection{Formulation of Forger}
\label{sec:RandomForger}
For an identity $j$, generation parameters
$\langle m, e, r \rangle$, and an object $I$, write
\begin{equation}
S^{m,e,I}_{j}(r) \;:=\;
Sim\bigl(L_{m}\bigl(e,\,(G \circ f_{j})(\langle m,e,r\rangle)\bigr),\, I\bigr),
\label{eq:forger-score}
\end{equation}
the similarity between $I$ and the object generated under identity $j$
with those parameters.
Let us consider Game~\ref{game:forge}, which describes how a forger may forge $\widetilde{\kappa}$ for an object $I$.

\begin{game}[Forgery Game]
\label{game:forge}
Fix a target random-forger success rate $p_r \in (0,1)$ and a query budget $Q$.
For an identity $i \in \mathcal{I}$ and a PPT adversary $\mathcal{A}$ that may
evaluate any $f_k \in \mathcal{F}$, at most $Q$ times in total, the game
$\mathbf{Game}^{\mathsf{forge}}_{\mathcal{A},i}(\lambda, Q)$ proceeds as follows:
\begin{enumerate}
  \item $\mathcal{A}$ chooses meta-parameters and an embedding $(m,e)$.
  \item The challenger samples $r \gets \{0,1\}^{n_r}$, sets
        $\kappa = \langle m,e,r\rangle$, and publishes the generated
        object $I = L_m\!\big(e,\, G(f_i(\kappa))\big)$ together with
        $(i,\kappa)$.
  \item $\mathcal{A}(I,i,\kappa)$ outputs a claim
        $(j, \widetilde{\kappa})$ with
        $\widetilde{\kappa} = \langle \widetilde{m}, \widetilde{e}, \widetilde{r}\rangle$.
  \item $\mathcal{A}$ wins iff the following conditions are met: (i) $j \neq i$, and (ii) $\Pr_{u \leftarrow \mathcal{U}}\left[ Sim\bigl(L_{\widetilde{m}}(\widetilde{e}, G(u)),\, I\bigr) > S^{\widetilde{m},\widetilde{e},I}_{j}(\widetilde{r}) \right] < p_r$, where $p_r$ is the small threshold chosen by the judge, where $S^{\widetilde{m},\widetilde{e},I}_j(\cdot)$ is as defined in \eqref{eq:forger-score}, and $u \leftarrow \mathcal{U}$ denotes uniform sampling from the set of all possible seeds. 
\end{enumerate}

\end{game}

In Game~\ref{game:forge}, the adversary knows the original $m$ and $e$ but is not restricted to using them. They are free to choose an alternative $\widetilde{m}$ and $\widetilde{e}$.
Up to this point, no assumption on the LDM is required. The only condition is that the acceptance probability ($p_r$) be sufficiently small, for example ($p_r \leq 2^{-50}$). The security argument relies on the rarity of the acceptance event: because the acceptance region has probability mass $p_r$, a forger should be unable to steer the PRF output into this region with probability greater than $p_r$ per evaluation.
% Indeed, should an adversary be able to produce such a forgery probability better than $p_r$ per evaluation, they would win  Game \ref{game:forge} with the same probability.
This probability is quantified in Definition \ref{definition:adv_forger}.   

% Using the following sequence of games, we show that any adversary capable of doing so could be transformed into an algorithm that breaks Assumption A1 for the underlying PRF. Therefore, under Assumption A1, no such adversary exists.
% Note that Step (4) of Game~\ref{game:forge} is evaluated by the
% ideal PA, which
% computes the exact tail probability
% $ \Pr[S^{m,e,I}_i(r) > S^{\widetilde{m},\widetilde{e},I}_j(\widetilde{r})]$. We first establish security in this idealized
% setting, where no assumption on the LDM is required. 

\begin{definition}[Forger's Advantage] \label{definition:adv_forger}
The advantage of $\mathcal{A}$ w.r.t.\ Game~\ref{game:forge} is
\[
\mathbf{Adv}^{\mathsf{forge}}_{\mathcal{A},i}(\lambda, Q)
:= \Pr\big[\mathcal{A} \text{ wins } \mathbf{Game}^{\mathsf{forge}}_{\mathcal{A},i}(\lambda, Q)\big],
\]
where the probability is over the challenger's choice of $r$ and the
randomness of $\mathcal{A}$.
\end{definition}

In other words, $\mathbf{Adv}^{\mathsf{forge}}_{\mathcal{A},i}(\lambda, Q)$
is the probability that a forger with query budget $Q$, who chooses the $m$ and $e$ used to generate an object $I$ under an author's identity $i$, and knows $\kappa = \langle m,e,r\rangle$, produces an alternative $(j, \widetilde{\kappa}) \neq (i, \kappa)$ that passes verification. That is, the forger is able to come up with an alternative set of generation parameters or identifier which is able to generate an object sufficiently similar to $I$. 
Since Game \ref{game:forge} quantifies over adversarially chosen $(m,e)$, an adversary could select a set of favourable parameters. Security thus must rely on neither the
secrecy nor the favourable choice of any generation parameter.%, but rather the randomness of the starting point induced by the PRF. 

The goal is to show that no PPT adversary outperforms the \textit{random-forger}---that is, the baseline that, after choosing $m$ and $e$, simply samples $\widetilde{r}$ uniformly at random---by a non-negligible probability in Game~\ref{game:forge}.
In particular, we show in Theorem \ref{thm:poa_PPT} that any non-negligible advantage over this random-guessing baseline would imply a violation of Assumption~A1 for the underlying PRF.

\subsubsection{Impossibility of a Polynomial-time Forger} \label{subsubsection:impossibility}

Under Assumption A1 (security of the PRF), no probabilistic polynomial-time  (PPT) forger can achieve a non-negligible advantage over a random-forger (with respect to the security parameter of the PRF family). This result is stated in Theorem~\ref{thm:poa_PPT}. The full proof is deferred to Appendix \appref{appendix:poa_PPT}.
% , and we present a sketch here. 
\begin{theorem} \label{thm:poa_PPT}
Under Assumption A1, for every identity $i \in \mathcal{I}$, query budget
$Q$ and random-forger success rate $p_r$ (as defined in
Game~\ref{game:forge}), and every PPT adversary $\mathcal{A}$ with query
budget $Q$,
\[
\mathbf{Adv}^{\mathsf{forge}}_{\mathcal{A},i}(\lambda, Q)
\;\le\; \mathrm{negl}(\lambda) + Q \cdot p_r,
\]
where $\lambda$ is the security parameter of the family $\mathcal{F}$.
\end{theorem}
Theorem~\ref{thm:poa_PPT} states that any forger achieves at most negligible advantage over a random forger: the $\mathrm{negl}(\lambda)$ term follows from the computational security of Assumption A1, and since a single
random-forger attempt succeeds with probability at most $p_r$, $Q$ attempts succeed
with probability at most $Q \cdot p_r$.

{%\color{blue}
\subsubsection{Adaptive and Optimization-Based Forgery Attempts} \label{subsubsection:adaptive_forger}
A natural question is whether a forger may employ an optimization-based or
adaptive strategy to construct a valid $\widetilde{\kappa}$ and outperform
the random-forger. The answer is no: Theorem~\ref{thm:poa_PPT} covers all
PPT forgers. Under Assumption~(A1), by
correlation intractability, no PPT algorithm with full knowledge of the
function family $\mathcal{F}$ can derive an input to any $f_j \in
\mathcal{F}$ whose output lands in the acceptance region better than random sampling. If such an adversary exists, it would, in effect, have identified structure
in $\mathcal{F}$, contradicting Assumption~(A1). This is the content of the
reduction in Theorem~\ref{thm:poa_PPT}.
}

% \begin{proof}[\emph{Proof Sketch}] 

% We prove this claim via a cryptographic reduction. 
% We first formalize the two properties of Section~\ref{subsection:prf}
% as cryptographic games: (1) collision resistance as
% Game~\ref{game:collision}, whose wins we show
% (Lemma~\ref{lemma:no_collision}) yield collisions of the underlying
% function,  and (2) correlation intractability as
% Game~\ref{game:steering}, whose target relation we show
% (Lemma~\ref{lemma:correlation_intractable}) is sparse by construction. Thus, by Assumption A1, no PPT adversary wins either Game \ref{game:collision} or Game \ref{game:steering} with non-negligible probability. 
% Next, we prove
% (Lemma~\ref{lem:forgery_implies_cr_ci}) that any PPT adversary
% $\mathcal{A}$ winning Game~\ref{game:forge} with non-negligible probability can be used to construct
% a PPT adversary winning at least one of Games~\ref{game:collision}
% or~\ref{game:steering} with non-negligible probability, violating Assumption A1. Hence,  no such $\mathcal{A}$ exists, and the claim follows. 

% \end{proof}
{

}

\subsection{Malicious Model Provider}
\label{subsubsection:collusion}
Could a malicious model provider collude with a forger with complete knowledge/white-box access to the model, e.g., via a model
backdoored to output highly similar objects under certain prompts regardless
of starting point? This failure arises when there exists
$S \subseteq \mathrm{supp}(\mathcal{D}_0)$ of non-negligible measure $c$ such
that $L_m(e,s)$ is highly similar to $I$ for every $s \in S$:
\[
\Pr_{s \sim \mathcal{D}_0}\big[\mathit{Sim}(I, L_m(e, s)) \ge k_0\big]
\;\ge\; \mu(S) \;=\; c,
\]
for some score level $k_0$ comparable to genuine claims, where $\mu$ is the probability measure induced by $\mathcal{D}_0$. We assume that the cryptographic function family $\mathcal{F}$ is secure; $f_i$ binds every claim
$\langle m,e,r \rangle$ to a starting point indistinguishable from a fresh
sample of $\mathcal{D}_0$. Let us refer to the event where $c$ is non-negligible as a \emph{concentration phenomenon}. Since $m,e$ are revealed to the PA, adversary and
PA sample from the same distribution. Thus, if the probability $c$ is negligible, the attacker cannot
trigger them with non-negligible probability. Conversely, if a concentration phenomenon occurs, the PA's samples trigger the adversarial outputs, yielding a large $\widehat{q}$ and
rejection. Either way, the collusion confers no advantage, as the backdoor faces an
inherent trade-off between being triggerable by the adversary, and being detectable by the PA.

\subsection{Colluding PA} \label{subsection:colluding_PA}
We consider the case where the PA itself is malicious. Since PA's assertion rests on Monte-Carlo sampling (Section~\ref{subsubsection:PA}), if this randomness were unconstrained, a PA colluding with a forger could bias the outcome while appearing to follow the protocol. For example, re-running the adjudication and retaining only samples favourable to the forger, or searching for a seed whose samples bias the outcome towards the forger's advantage. Nonetheless, such manipulation can be precluded deriving the PA's randomness from a cryptographic hash of its inputs, making the entire adjudication a deterministic function of the publicly available parameters. This is similar to a smart contract, fixed by its inputs and re-executable by anyone. Thus, any search for favourable generation parameters reduces to the random-forger strategy of Section~\ref{sec:RandomForger}.

% \subsubsection{Verification Failure}
% The concentration phenomenon described in Section \ref{subsubsection:collusion} can also arise benignly, there could be genuine instances where an honest model and author cannot
% certify. For example, the training data may be insufficient for certain prompts, or overfitted during training, leading limited output diversity~\cite{gu2025memorizationdiffusionmodels}. In this case, the
% model generates similar objects regardless of starting point, and the author
% cannot certify. Nonetheless, the author can still detect this in advance by
% running the public PA procedure and switching prompt or model if
% $\widehat{q}$ is not small.

\subsection{Adversarial Attack on Similarity Function} \label{subsection:adv_attack}
Given a contested object $I$, generated from latent $\mathcal{L}$, and its corresponding generation parameters $\kappa = \langle m,e,r\rangle$, consider a distance function $D(\cdot, \cdot)$ over the metric space of $\mathcal{L}$, such as the  $\ell_2$ distance.  An adversary may attempt an adversarial attack on $Sim(\mathcal{L}, \cdot)$, by finding a $\widetilde{\mathcal{L}}$ s.t. $D(\mathcal{L}, \widetilde{\mathcal{L}})$ is small, but  $Sim(\mathcal{L}, \widetilde{\mathcal{L}})$ is significantly reduced. This causes the judge to fail to attribute  $\widetilde{\mathcal{L}}$  to $\kappa$. 

\subsubsection{Robustness against Adversarial Attack in $\ell_p$ }
Suppose  the adversary wants to construct a $\widetilde{\mathcal{L}} = \mathcal{L} + v$, where $||v||_p \leq \epsilon$, such that the PA computes a low similarity score and high $\hat{q}$, indicating that $\widetilde{\mathcal{L}}$ is not generated by author $i$ using generation parameters $\kappa$.
For any $p\geq 1$ and $||v||_p \leq \epsilon$, we show in Appendix~\appref{appendix:adversarial} that  $Sim(\mathcal{L}, \widetilde{\mathcal{L}}) \geq Sim(\mathcal{L}, \mathcal{L}) - \frac{1}{d}\epsilon ||\mathcal{L}||_q$, where $q = \frac{p}{p-1}$. For example, when $p=2$, $q=2$ whereas when $p=1, q = \infty$. Furthermore, for all $q>1$ (i.e., $p<\infty$), $\|\mathcal{L}\|_q \in o(d)$. Therefore, the impact of $v$ on the similarity score, for reasonable choices of $\epsilon$ which preserve perceptual similarity between the generated objects, is negligible. 

Since we impose no structural assumptions on the perturbation $v$ beyond the norm constraint $||v||_p \leq \epsilon$, the above bound characterizes the worst-case degradation of the similarity score under $\ell_p$-bounded adversarial perturbations. We therefore conclude that the similarity function is robust against $\ell_p$ adversarial perturbations that aim to reduce it (including gradient-based optimizations such as FGSM or PGD).

\subsubsection{On Geometric Transformations}
To evade detection, the forger may attempt a perceptual preserving geometric transformation on the generated object. Unlike $\ell_p$-bounded perturbations, geometric distortions induce non-local but structured changes.  In this case, a transformation of the forged object $\mathcal{L}$ may correspond to a transformed version of the original $L_m(e, f_i(\kappa))$. For example, the forged object could be a rotated or scaled version of the original. To mitigate this, the author provides a transformation $t$ as described in Section \ref{subsection:PA_spec}, to align the generated and forged objects prior to comparing similarity. These transformations are also evaluated by the judge to determine if they are reasonable. 
Designing similarity functions robust to such  distortions remains an open area of research.

{%\color{blue}
\section{Motivation and Validity of Assumption A2} \label{section:A2-analysis}
In Section~\ref{sec:assumption}, we introduced Assumption A2. Intuitively, the similarity
score is an aggregate statistic over a large number of latent
coordinates, so its sub-exponential distribution arises naturally from
probability concentration and hence, Assumption A2 is met. 
To show that A2 indeed
arises from the structure of the LDM, we introduce the following
condition, which we then show is sufficient to establish A2:
\begin{enumerate}
 \item[(A2-I)] Individual latent coordinates are sub-Gaussian.
\end{enumerate}
Lemma~\ref{lemma:subexp}
(Section~\ref{subsection:similarity}) states this result: under
Assumption~A2-I, the similarity score admits a sub-exponential tail,
and hence Assumption~A2 is satisfied. The following are evidence supporting (A2-I):
\begin{enumerate}
    \item It is a mild relaxation of the idealized VAE prior, under
    which the unconditional latent coordinates are optimized towards
    $\mathcal{N}(0,1)$ as a VAE training objective
    (Section~\ref{subsubsection:vae}).
    \item (A2-I) is routinely adopted in theoretical analyses of
    generative models: as an assumption on diffusion models
    \cite{zhang2024minimax, wang2024wasserstein}, and on the relevant
    loss or data in generalization and excess-risk analyses of VAEs
    \cite{chen2025generalization, tang2021empirical}. Moreover, (A2-I)
    is implied by, and hence weaker than, the bounded-support
    assumptions common in prior analyses
    \cite{mbacke2023pac, oko2023diffusion}, since any bounded random
    variable is sub-Gaussian \cite[Example 2.6.5]{Vershynin2026HDP}.
    \item It is empirically verifiable, as illustrated in
    Section~\ref{section:sub-exp-tail}.
\end{enumerate}
%  In our POA framework, the seed is constructed by passing the author's unique identifier and
% generation parameters into a PRF $f_i$. 
% The seed is then passed into a pseudorandom number generator $G$ to generate an initial random starting point
% for the diffusion process according to some distribution (e.g. multivariate Gaussian distribution).  

\subsection{Validity of Assumption A2} \label{subsection:similarity}

% Assumption~A2 requires the similarity score between a fixed
% contested object and an independently generated object to have a
% sub-exponential tail.
Lemma~\ref{lemma:subexp} provides a sufficient condition for
Assumption~A2 by showing that the similarity statistic
\begin{equation} \label{eqn:T}
\operatorname{Sim}(\mathcal{L},\mathcal{L}')
:=
T
=
\frac{1}{d}\mathcal{L}^{\top}\mathcal{L}'
\end{equation}
is sub-exponential whenever the coordinates of \(\mathcal{L}\) and
\(\mathcal{L}'\) satisfy the sub-Gaussian condition in
Assumption~(A2-I). 
\begin{lemma}[Sub-exponential tail of the similarity score]
\label{lemma:subexp}
Let $T$ be as defined in \eqref{eqn:T}.
Under Assumption~\textnormal{(A2-I)}, there exists constants $c, K>0$
% there exist universal constants
% \(c,C>0\) such that
% \[
% \left\|
% T-\mathbb{E}[T]
% \right\|_{\psi_1}
% \leq CK^2,
% \]
% where $\| \cdot \|_{\psi_p}$ denotes the Orlicz norm for $\psi_p(x) = \exp(x^p)-1$.
% Consequently, 
for every \(t\geq 0\),
\[
\mathbb{P}\left(
\left|T-\mathbb{E}[T]\right|
\geq t
\right)
\leq
2\exp\left(
-c\frac{t}{K^2}
\right) \in \exp(-\Omega(t)),
\]
i.e. the tail of the probability distribution is upper bounded by a sub-exponential tail. 
\end{lemma}

In particular, independently generated latents attain 
large similarity scores  with exponentially decreasing
probability. Lemma~\ref{lemma:subexp} therefore provides a
sufficient theoretical justification for Assumption~A2. The full
proof is deferred to Appendix~\appref{appendix:proofs_similarity}.

\begin{remark}[\emph{A stronger dimension-dependent tail bound}] \label{remark:subexp_bernstein}
% In Lemma \ref{lemma:subexp}, we showed that the tail of \(T\) decays no slower than a sub-exponential function, i.e.
% \[
% \Pr\left(|T-\mathbb{E}T|>t\right)
% \leq
% \exp\bigl(-\Omega(t)\bigr).
% \]
Under an additional weak-dependence assumption on
\(\mathcal{L}\), Bernstein-type arguments can yield a stronger,
dimension-dependent tail bound of order
\[
\Pr\left(|T-\mathbb{E}T|>t\right)
\leq
\exp\bigl(-\Omega(d)t\bigr),
\]
 see Appendix~\appref{app:stronger-tail} for more details.
The bound in Lemma~\ref{lemma:subexp}, however, is sufficient to justify
Assumption~A2 with minimal conditions.
\end{remark}

\begin{remark}[\emph{Independent honest generations}] \label{remark:independent_honest_generation}
Assumption~A1 also bounds accidental collisions between honest authors.
Suppose two authors independently generate objects under their respective
identities and generation parameters. The acceptance region for any fixed
object has probability at most $p_r$ under uniformly random seeds, by
definition of the judge's threshold, so by correlation intractability
(A1) the second author's single evaluation of $f_j$ lands in it with
probability at most $p_r + \mathrm{negl}(\lambda)$. The ambiguous scenario
in which two honest authors both hold accepted claims on the same object
therefore arises with probability at most $p_r + \mathrm{negl}(\lambda)$.

\end{remark}
}
\subsection{Adjudication with Assumption A2}
\label{subsection:statistical-adjudication}

Having established that unusually large similarity scores are rare
under the random-forger distribution, we next explain how the PA
uses an observed score to adjudicate an authorship claim.
Theorem~\ref{theorem:hypothesis} shows that thresholding the
similarity score yields a statistically optimal test.\footnote{The
test is statistically optimal in the sense that it is uniformly most
powerful: among all tests with a fixed false-positive rate, it
achieves the lowest possible false-negative rate.}

By Lemma~\ref{lemma:subexp}, the similarity statistic admits a sub-exponential tailed probability distribution. For tractable
empirical estimation, we model its distribution using a generalized
normal distribution, which admits a flexible parametric sub-exponential tailed probability distribution
\cite{nadarajah2005generalized,zhang2020generalized,aljarrah2019new}.
This generalized-normal approximation is an additional modeling
choice, rather than a direct consequence of
Lemma~\ref{lemma:subexp}, and is evaluated empirically in
Appendix~\appref{section:expt_sim}.

A generalized normal distribution is parameterized by a location
parameter $\mu$, a scale parameter $\gamma$, and a shape parameter
$\beta$, which may be jointly estimated using maximum likelihood methods \cite{nadarajah2005generalized}, 
under standard regularity conditions, their estimation error
decreases at the rate $O(n^{-1/2})$ \citep[Theorem 5.39]{Vaart_1998}.

\begin{theorem} \label{theorem:hypothesis}
Let $T = \mathrm{Sim}(\mathcal{L}, \mathcal{L}') = \frac{1}{d}\mathcal{L}\cdot\mathcal{L}'$
denote the similarity score between $\mathcal{L}, \mathcal{L}'$. Assume that
$T$ follows a generalized normal distribution with location $\mu$ and known
scale and shape parameters $\gamma > 0$ and $\beta \ge 1$. Consider the
hypothesis test $H_0 : \mu = 0$ ($\mathcal{L}$ and $\mathcal{L}'$ are not
similar) versus $H_1 : \mu > 0$ ($\mathcal{L}$ and $\mathcal{L}'$ are
similar), and the test $\phi(T) = 1(T \ge W_a)$, where $W_a$ is the
$(1-a)$ quantile of $W$ (as defined in Section~\ref{subsubsection:PA}) given
$I$. Then $\phi$ has size $a$ and is uniformly most powerful among all
tests with false-positive rate at most $a$.
\end{theorem}

The restriction $\beta \ge 1$ is a consequence of Assumption (A2). 
 When the similarity scores are centered under the null,
$\mu_0=0$. In the implemented PA, the distribution parameters and
the null quantile are estimated from independent Monte Carlo
samples. The resulting procedure is therefore a plug-in
approximation to the idealized test in
Theorem~\ref{theorem:hypothesis}, with finite-sample estimation
error accounted for through the confidence interval described in
Section~\ref{subsection:PA_spec}.
{%\color{blue}
\subsection{Adjudication Without Assumption A2}
\label{sec:distribution-free}

The security analysis in Section~\ref{sec:securityanalysis} does not
require Assumption A2. It relies on Assumption A1 and on the judge
accepting only claims that correspond to a rare event under the
random-forger distribution. Assumption A2 is instead used to
efficiently estimate the probability of this event. As described in Section \ref{subsection:PA_spec},
the sub-exponential tail model allows the PA to certify very small
random-forger success probabilities using only a polylogarithmic
number of samples.

\subsubsection{Empirical Assessment of Assumption A2}

By Lemma~\ref{lemma:subexp}, Assumption A2 follows
from the sufficient condition A2-I: each coordinate of the generated
latent is sub-Gaussian. A random variable is sub-Gaussian if and only if
its Orlicz \(\psi_2\) norm is finite~\cite{Vershynin2026HDP}. In
Section~\ref{section:sub-exp-tail}, we empirically assess A2-I by
examining the estimated \(\psi_2\) norm and, in particular, whether
the estimate stabilizes as the number of sampled generations increases.

These diagnostics provide empirical support for A2-I, but do not
constitute a proof of sub-Gaussianity, since finite samples cannot
exclude unobserved heavy-tail behavior. Moreover, since A2-I is a
sufficient but not necessary condition for A2, a failure to empirically
support A2-I does not by itself imply that A2 is false.

\subsubsection{Distribution-Free Fallback}

If Assumption A2 cannot be adequately supported, the PA can avoid
parametric tail extrapolation and directly estimate the
random-forger success probability. Let
\[
X_k =
\mathbf{1}\!\left\{
    Sim(I,L_m(e,S_k)) \geq T
\right\},
\qquad S_k \leftarrow D_0,
\]
so that $\mathbb{E}[X_k]=q$. The PA may obtain a distribution-free
upper confidence bound on $q$ from the empirical $(1-q)$ quantile of the Bernoulli samples
$X_1,\ldots,X_n$.

For example, if no exceedance is observed in $n$ independent
samples, an exact $(1-\delta)$ upper confidence bound is
\[
q \leq 1-\delta^{1/n}
    \leq \frac{\ln(1/\delta)}{n}.
\]
Consequently, certifying $q \leq q_0$ with confidence $1-\delta$
requires
$
n \geq 
    \frac{\ln(1/\delta)}{q_0}
$
samples. This remains practical for moderate values of $q_0$ such as
$2^{-10}$ or  $2^{-20}$, but is infeasible for values such as $2^{-50}$.

Nonetheless, the security guarantees as a result of Assumption A1 continue to hold.  For any fixed claimed
generation parameters, a PPT adversary cannot obtain a
non-negligible advantage over the corresponding random-forger
distribution. However, without A2, the efficiently certifiable
random-forger baseline may be substantially larger. If the model is
backdoored, highly memorizing, or otherwise maps a non-negligible
fraction of starting points to similar objects, the empirical
procedure will reveal a larger value of $q$, and the judge must
either reject the claim or provide only a weaker certificate.

}

\section{Evaluations} \label{section:empirical}

 We empirically verify the following:
 \begin{enumerate}
    \item Section~\ref{section:sub-exp-tail}: For a fixed prompt $e$ and model
parameters $m$, the generated latents $L_m(e, s_1), L_m(e, s_2), \dots$ generated from
independently sampled starting points $s_1, s_2, \dots$ are sub-Gaussian and
exhibit weak inter-coordinate dependence.

\item Section~\ref{subsection:tail_validation}: Since validating the
$2^{-50}$ tail directly is infeasible, we compare the fitted
distribution on $2^{10}$ samples against fitted tail frequencies at
survival probability $2^{-17}$ (w.r.t the empirical quantile estimated using $2^{17}$ samples), showing the
certified upper confidence bound remains conservative at this depth.
     % \item Section \ref{section:expt_sim}: Distribution of $Sim(L_m(e, s_1), L_m(e,s_2)), \dots$, where $s_1, s_2, \dots \sim \mathcal{D}_0$, follows a sub-exponential distribution as described in Assumption A2, and can be accurately approximated using small number of samples
    \item Section \ref{section:similarity_eval}: The probabilistic adjudicator (PA) described in Section~\ref{subsection:PA_spec} is able to accurately verify assertions with high statistical confidence, determining whether a contested object corresponds to a particular author and POA. 
    \item Section \ref{section:eval_robust}: The PA is empirically shown to be robust to noise (e.g., additive Gaussian noise and JPEG compression) and perceptually preserving transformations (e.g. cropping, rotation, scaling, translation).
 \end{enumerate}
 Additionally, the following evaluations are deferred to Appendix \appref{appendix:additional_eval}: Section \ref{section:expt_distance} empirically verifies $\ell_2$ distance between generated latents is preserved up to a constant factor relative to their starting points, Section \ref{section:expt_sim} empirically evaluates the goodness-of-fit of the fitted GND distribution. Finally, 
 Appendix~\appref{appendix:samples_generated} presents representative images generated using our POA scheme.

\subsection{Experimental Setup} \label{section:expt_setup}

\paragraph{LDM and datasets}
We employed Stable Diffusion 2.1 \cite{DBLP:conf/cvpr/RombachBLEO22} and Stable Diffusion XL \cite{podell2024sdxl} with their default settings as downloaded from HuggingFace. 
 Images were generated using prompts sourced from the Stable Diffusion Prompt dataset\footnote{\url{https://huggingface.co/datasets/Gustavosta/Stable-Diffusion-Prompts}}. Finally, we employed the HMAC-SHA3 hash as our PRF to construct the seed using the  author's identifier and generation parameters, which is then passed into PyTorch's Gaussian random generator to sample from $\mathcal{N}(0, 1)$.

\paragraph{System Configuration}
Processor: Intel(R) Xeon(R) CPU E5-2620 v4 @ 2.10GHz;
RAM: 256GB;
OS:  Ubuntu 20.04.4 LTS;
GPU: NVIDIA Tesla V100.

% \subsubsection*{Evaluation Metrics} 
% To evaluate the effectiveness of our proposed POA, we consider the following metrics:
% \begin{enumerate}
%     \item \emph{False positive rate (FPR)}: Proportion of images wrongly attributed to an incorrect POA
%     \item \emph{Accuracy (ACC)}: Proportion of images correctly linked with their corresponding POA. 
% \end{enumerate}

\paragraph{Registration, Generation and Contention}
The registration, generation and contention process is described in Section \ref{subsubsection:generation}. The specifics of our implementation is described below. 

\paragraph{\textit{Registration}} An author is assigned a unique randomly sampled integer identifier $i$. 

\paragraph{\textit{Generation} } The author experiments on various LDM parameters $m$, prompt embedding $e$, and a sequence of random bits $r$. Next, compute the seed $s_{\langle m, e, r \rangle} = f_i (\langle m, e, r \rangle)$, where $f_i$ is the HMAC-SHA3 pseudorandom function using $i$ as the key. Finally, the author generates the starting point $G(s_{\langle m, e, r \rangle})$, and uses it to generate the latent space representation $\mathcal{L}$, which is then upscaled into the final generated image $I$ through the VAE. In our evaluations, we employed PyTorch's Gaussian random number generator as $G$, and the default VAE included in the Stable Diffusion 2.1 pipeline. 

\paragraph{\textit{Contention} } Suppose the author wishes to assert their authorship of an image $I'$ using the generation parameters $ \kappa = \langle m, e, r \rangle$. The following is performed:
\begin{enumerate}
    \item Generate the latent space representation $\mathcal{L}'$ by passing $I'$ into the LDM's VAE encoder.
    \item Generate  $\mathcal{L} = L_m(e, G(f_i(\kappa)))$. 
    \item 
    Compute the test statistic in Theorem \ref{theorem:hypothesis} over $\mathcal{L}$ and $\mathcal{L}'$, which is $T = \frac{1}{d}\mathcal{L}^T \mathcal{L}'$. 
\end{enumerate}

{%\color{blue}
\subsection{Distributional Properties of Generated Latents} \label{section:sub-exp-tail} 

\subsubsection{Sub-Gaussian Tails (Assumption A2-I)} A sufficient condition for Assumption A2 is Assumption A2-I: the generated latent coordinates are sub-Gaussian, or equivalently, have finite sub-Gaussian (\(\psi_2\)) norm. We first assess this condition on Stable Diffusion 2.1 (SD 2.1) by examining whether the empirical \(\psi_2\)-norm estimate stabilizes as the number of independent generations increases. For each prompt, we generate \(8192\) latents and compute the empirical norm over progressively larger subsets of generations. At each sample size, we perform a generation-level bootstrap with \(100\) resamples, resampling entire latents with replacement to preserve within-latent dependence, and report the resulting \(95\%\) bootstrap interval. 

\begin{figure} \centering \includegraphics[width=1\linewidth]{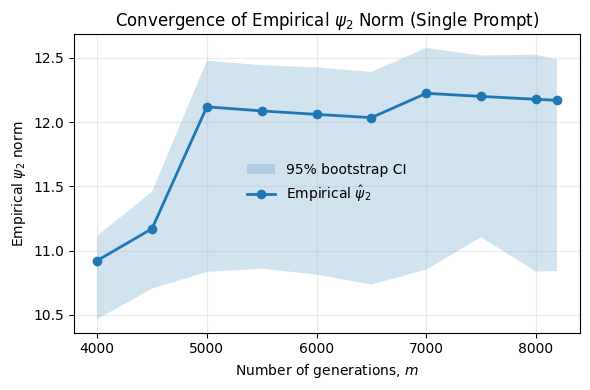} \caption{Empirical \(\psi_2\)-norm estimate for SD 2.1 as the number of independent generations increases for a representative randomly sampled prompt. The shaded region denotes the \(95\%\) generation-level bootstrap interval computed from \(100\) resamples. The estimate stabilizes as the sample size increases.} \label{fig:psi2_convergence} \end{figure} 

Figure~\ref{fig:psi2_convergence} shows the SD 2.1 result for a  randomly sampled prompt. The estimate initially increases but stabilizes at approximately \(12.1\) beyond \(5000\) generations; at \(8192\) generations, \(\hat{\psi}_2 = 12.17\), with a \(95\%\) bootstrap interval of \([10.84,12.49]\). We observe similar stabilization across all evaluated SD 2.1 prompts, with final empirical \(\psi_2\)-norm estimates ranging from \(9.4\) to \(17.5\). We repeat the same experiment on SD XL and observe
the same stabilization behavior, with final estimates ranging from
\(11.9\) to \(14.2\). Taken together, these results provide empirical evidence consistent with
a finite \(\psi_2\) norm for both SD 2.1 and SD XL, thereby providing
empirical support for Assumption (A2-I).

\subsubsection{Weak Inter-Coordinate Dependence (Assumption (A2-II))}
We have earlier established (Remark~\ref{remark:subexp_bernstein}) that (A2-II) is not
required for Assumption~(A2), which follows from (A2-I) alone
(Lemma~\ref{lemma:subexp}). However, under weak inter-coordinate
dependence, the tail of the similarity score sharpens to
\(\exp(-\widetilde{\Omega}(d)t)\) (Remark~\ref{remark:subexp_bernstein}).
The resulting decay in the latent dimension \(d\)
(\(d=4\times64\times64\) for SD 2.1) provides stronger support for
certifying very small forgery probabilities such as \(2^{-50}\).

To empirically assess (A2-II), for each prompt we generate \(n=100\)
latents from independently sampled starting points and estimate the
\(d\times d\) correlation matrix of the latent coordinates across
generations. Even under zero true correlation, finite-sample noise produces
nonzero empirical correlations; at \(n=100\), the expected magnitude can be shown to be
$
\sqrt{\frac{2}{\pi n}} \approx 0.08,
$
 see \cite{fisher1915frequency} for its derivation. 
Across \(100\) randomly sampled prompts, at least \(80\%\) of the
off-diagonal correlations in SD 2.1 and \(84\%\) in SD XL have magnitude
below \(0.1\), close to the finite-sample noise level. Furthermore,
correlations are concentrated among spatially neighboring coordinates and
decrease as the spatial separation between coordinates increases. Taken
together, these observations provide empirical evidence consistent with
predominantly weak and localized inter-coordinate dependence in both
models, supporting Assumption (A2-II).

\subsection{Tail Behavior of the Similarity Score Under the Fitted Distribution}
\label{subsection:tail_validation}
In Section~\ref{subsection:PA_spec}, we described how Assumption~A2 allows the PA
to certify a random-forger success rate $q$ using polylogarithmically many
samples instead of $1/q$. The practical concern, however, is sharper than
the tail bound of Section~\ref{subsection:similarity}: the PA fits the
distribution on bulk samples and extrapolates into an unobserved tail
region, and the reliability of this extrapolation must be assessed
empirically. Since direct validation at $q = 2^{-50}$ would require on the
order of $2^{50}$ samples, we validate at the deepest computationally
feasible depth, $q = 2^{-17}$, comparing a GND fitted on $2^{10}$ samples
(the PA's sample size) against empirical tail frequencies from $2^{17}$
independent samples:

\begin{enumerate}
    \item Sample a random prompt embedding $e$. Fix a reference object
          $X_0 = L_m(e, s_0)$, and generate
          $X_i = L_m(e, s_i)$, $i = 1, \dots, 2^{10}$, and
          $Y_j = L_m(e, s'_j)$, $j = 1, \dots, 2^{17}$, with all starting
          points sampled independently from $\mathcal{D}_0$.
    \item Fit a GND by maximum likelihood to
          $S_{\mathrm{fit}} = \{\mathit{Sim}(X_0, X_i)\}_{i}$, with fitted
          survival function $Q(x) = 1 - F(x)$. Let $\widehat{Q}$ denote the
          empirical survival function of
          $S_{\mathrm{val}} = \{\mathit{Sim}(X_0, Y_j)\}_{j}$.
    \item For each $x$ in the deepest observable tail slice
          $\mathcal{V} = \{\, x \in S_{\mathrm{val}} :
          2^{-17} \leq \widehat{Q}(x) \leq 2^{-16} \,\}$, compute
          $
              \rho(x) \;=\; \log_2 \frac{Q(x)}{\widehat{Q}(x)}.
          $
\end{enumerate}

The procedure is repeated over 100 randomly sampled prompt embeddings, and
we report the distribution of $\rho$ aggregated across prompts and
thresholds. A value $\rho > 0$ means the fitted model overestimates the
random-forger's success probability---the conservative direction, making
forgery harder---while $\rho < 0$ indicates anti-conservative
extrapolation that may understate the true rate.

Finally, recall that the PA certifies not the point estimate $\hat{q}$ but
the upper endpoint of $\mathcal{C}_\delta$
(Section~\ref{subsection:PA_spec}). The former reflects the quality of the
GND extrapolation itself; the latter is the quantity actually used by the
judge, and the gap between them is the safety margin added by the
confidence interval. Table~\ref{tab:tail_validation} reports $\rho$ across
the 100 prompts under both choices, and
Figure~\ref{fig:rho_single_prompt} shows the corresponding distributions
for a single randomly chosen prompt. Together they show that while the
point estimate can be anti-conservative for individual prompts, the
certified bound consistently overestimates the forger's success
probability, erring on the side of rejecting forgeries. While this
validation extends only to the $2^{-17}$ quantile, it provides supporting
evidence that the extrapolation underlying the PA's certificate is
reliable.
}
\begin{table}
\centering
\caption{Distribution of the tail log-ratio $\rho$ across 100 prompt
embeddings (per-prompt averages over the tail slice $\mathcal{V}$), computed
using the fitted point estimate and the conservative upper confidence bound (to 1 decimal place).}
\label{tab:tail_validation}
\begin{tabular}{lccc}
\toprule
Model tail probability & Mean & Min & Max \\
\midrule
Point estimate $Q(x)$        & 0.7 & $-1.9$ & 1.1 \\
Conservative bound           & 3.2 & 2.0    & 4.0 \\
\bottomrule
\end{tabular}
\end{table}

\begin{figure}
    \centering
    \includegraphics[width=1\linewidth]{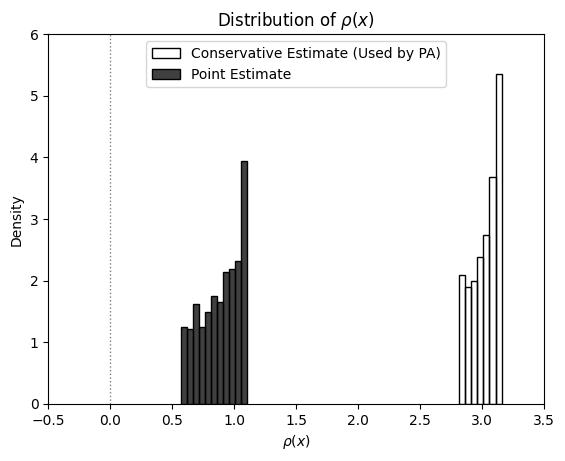}
    \caption{Distribution of $\rho(x)$ over the tail slice $\mathcal{V}$ for a
single prompt, using the PA's conservative bound and the
fitted point estimate. $\rho(x) > 0$ indicates a conservative
estimate.}
    \label{fig:rho_single_prompt}
\end{figure}

\subsection{Similarity Function Evaluation} \label{section:similarity_eval}

We evaluate similarity metric proposed in Section \ref{subsection:similarity} (i.e. $Sim(X, Y) = \frac{1}{d}X \cdot Y$) and the Probabilistic Adjudicator described in Section \ref{subsection:PA_spec} under the following settings:
\begin{enumerate}
    \item \textbf{Clean setting:} $X$ and $Y$ are clean and unmodified images generated directly by the LDM
    \item \textbf{Noise and compression setting:} $Y$ is obtained from $X$ by applying Gaussian noise or JPEG compression.
    \item \textbf{Affine transformation setting:}  $Y$ is obtained from $X$ via a mild, perceptually preserving affine transformation.
\end{enumerate}

\paragraph{Clean Performance and Robustness to Noise and Compression}
To evaluate the performance of our similarity metric and PA, we proceed as follows. For a $\delta = 10^{-4}$, contested object $I$, and $\kappa = \langle m,e,r\rangle$, we run the PA described in Section \ref{subsection:PA_spec} as follows:
\begin{enumerate}
\item Independently sample $s_1, \dots, s_{n} \sim \mathcal{N}(0, I)$, where $n = \ln^{2} \frac{1}{\alpha} \ln \frac{1}{\delta}$. Compute $s_0 = G(f_i(\langle m, e, r\rangle))$ as the starting point derived from the author's identity and generation parameters.  
\item Generate latents $J_0, \dots, J_{n}$, where $J_j = L_m(e, s_j)$ denotes the flattened latent generated conditional on $(m, e)$ and starting point $s_j$.

\item Compute similarity scores $Sim(J_0, J_j)= \frac{1}{d} J_0 \cdot J_j,$ for $j = 1, \dots, n$, 
as defined in Section~\ref{subsection:similarity}, where $d = 4 \times 64 \times 64$.

\item Fit a generalized normal distribution to the empirical similarity scores using the {gennorm} function from Python’s {scipy} library, and estimate the empirical similarity distribution $\widehat{W}$

\item Compute $T=Sim(J_0, I)$. Output $\hat{q} = \IP[\widehat{W} \geq T]$ and its corresponding $(1-\delta)$ confidence interval $\mathcal{C}_\delta$
\end{enumerate}

\paragraph{Judge}  Let $q'$
denote the upper quantile of $\mathcal{C}_{\delta}$. We simulate the judge's decision by accepting the assertion that the contested $I$ originates from the original $J_0$, if $q' \leq p_r$, where $p_r \in \{2^{-10}, 2^{-30}, 2^{-50}\}$, with a corresponding choice of $\alpha = p_r / 2$. Recall from Section \ref{subsection:PA_spec} that $p_r$ upper bounds the probability that a random-forger generates an object whose similarity score exceeds the threshold $T$, and $\alpha$ is an upper bound of the estimation error of $\hat{q}$ with probability at least $(1-\delta)$

This procedure is repeated over $100$ text embeddings generated from randomly sampled prompts. We consider the following types of distortions applied to $I_0$:
\begin{itemize}
    \item \textbf{Gaussian noise:} Add noise sampled from $\mathcal{N}(0, \sigma^2)$ to the base image, with $\sigma^2 \in \{1, 3, 9\}$.
    \item \textbf{JPEG compression:} The decoded image corresponding to $I_0$ is compressed using JPEG with quality ($q$) parameters $75, 50$ and $10$, where lower $q$ indicates higher compression.
\end{itemize}

\subsection{Robustness against Affine Transformations } \label{section:eval_robust}
To evaluate robustness against mild geometric distortions, we additionally consider a class of perceptually preserving affine transformations. Specifically, for each original $I_0$, we apply a randomly sampled affine transformation parameterized as follows:
\begin{itemize}
\item \textbf{Scaling (zoom):} A multiplicative scaling factor sampled uniformly from $[0.98, 1.02]$, corresponding to a mild zoom of up to $\pm 2\%$.
\item \textbf{Translation:} Horizontal and vertical translations sampled uniformly from $[-2\%, 2\%]$ of the image dimensions.
\item \textbf{Rotation:} A rotation angle sampled uniformly from $[-3^\circ, 3^\circ]$.
\item \textbf{Shearing:} A shear angle sampled uniformly from $[-2^\circ, 2^\circ]$.
\end{itemize}
In our evaluations, we assume that the PA can accurately determine the affine transformation used. Let $t$ be the transformation that aligns the contested $I$ and the original object $J_0$ (generated using the author's identifier and generation parameters).
We modify the testing procedure described in Section \ref{section:similarity_eval},  to account for affine perturbations as follows: 
\begin{itemize}

\item Replace Step~(3) with: Compute $Sim(J_0, t(J_j))$ for $j = 1, \dots, n$.

\item Replace Step~(5) with:  Compute $T=Sim(J_0, t(I))$. Output $\hat{q} = \IP[\widehat{W} \geq T]$ and its corresponding $(1-\delta)$ confidence interval $\mathcal{C}_\delta$
\end{itemize}

For each of the 48 combinations of model (SD~2.1 and SD~XL),
distortion setting (clean, Gaussian noise with
$\sigma^2 \in \{1,3,9\}$, JPEG compression with quality
$q \in \{75, 50, 10\}$, and random affine transformation), and
random-forger success rate $p_r \in \{2^{-10}, 2^{-30}, 2^{-50}\}$,
the judge is evaluated over 100 prompts. The false-reject rate is
$0.00$ in every configuration except one: JPEG compression at
$q=10$ on SD~2.1 with $p_r=2^{-50}$, where a single prompt out of
100 is rejected ($0.01$). Moreover, we did not encounter any false
positives across all settings. The similarity-score distributions
remain well separated even under substantial noise, compression,
and affine perturbation.

Note that the false-reject rate refers to the probability that an object with a valid generation parameters is incorrectly rejected by the judge at a given random-forger success rate. Furthermore, during our evaluations, we did not encounter any false positives (i.e. random images which were above the similarity threshold) across all settings. 
% The similarity-score distributions remain well separated and the empirical false-positive rates remain negligible even under substantial noise, compression, and affine perturbations.

{
%\color{blue}

\subsection{Robustness Against Regeneration Attacks} \label{section:eval_regen}
Watermark-removal attacks target a hidden signal: watermarking
schemes embed a secret-derived pattern into the generated object, and
verification depends on extracting this signal. Removal attacks
succeed by perturbing the object enough to destroy the embedded
signal while preserving perceptual quality. Our framework has no
such hidden signal to destroy: verification is based on the
similarity of the entire latent representation to the regenerated
original, rather than on the presence of an embedded signal.
Although the similarity comparisons by the PA occurs in the latent space,  latent-based watermark removal techniques such as~\cite{lee2025removal, hu2024stablesignatureunstableremoving}
do not directly apply, as they aim to disrupt watermark signals in the latent space

The closest analogue to a removal attack in our setting is a
\emph{regeneration attack}~\cite{NEURIPS2024_10272bfd}, in which an
adversary uses a diffusion model to re-synthesize a perceptually
similar version of the image. Regeneration has been shown to reliably
strip watermark signals, and it aligns with the adversary's goal in
our threat model (Section~\ref{sec:securityanalysis}): producing an object
that remains perceptually close to the original while suppressing the
statistical evidence linking it to the author's generation
parameters. In our framework, a successful attack corresponds to the
regenerated object falling below the similarity threshold, so that
the judge can no longer attribute it to the author.

We evaluate using the settings recommended
in~\cite{NEURIPS2024_10272bfd}, over the same configurations as in previous sections. The attack success rate — the fraction of
regenerated objects that evade attribution at
$p_r = 2^{-50}$ — is $0\%$ for both SD~2.1 and SD~XL. That is, in every case,
the regenerated object remains sufficiently similar  
for the PA to attribute it to the original generation parameters.
}

\section{Related Work } \label{section:related_work}

\subsection{Digital Watermarking for Authorship Protection} \label{subsection:related_watermark}
Image watermarking has been a topic of extensive research for several decades, 
with work as early as 1994 \cite{413536}. Early watermarking schemes focus on the post-watermarking setting, where watermarks are embedded in the pixel or frequency domain (e.g. \cite{DBLP:conf/spieSR/KutterJB97, DBLP:conf/mm/KankanhalliRR98}). The latter has been shown to be more robust against mild affine and geometric perturbations. Such watermarking schemes have been used to embed authorship information into digital media, see for example \cite{sheppard2004secure, DBLP:journals/access/MalanowskaMAMK24}.

Recently, there has been growing interest in watermarking AI-generated content, such as the output of large language models (LLMs) and diffusion models.  \cite{DBLP:conf/icml/KirchenbauerGWK23} introduced the notion of undetectable watermarks in LLMs, which embeds the watermark by controlling the model's pseudorandom number generator. Since then, there has been continued interest in watermarking AIGC, for use cases such as copyright protection and misuse tracking. 

Watermarking techniques for AIGC created using LDMs can be broadly categorized as follows. 
\begin{enumerate}
    \item \textbf{Post-processing watermarking}: Embeds watermarks directly into generated images using traditional pixel and frequency domain, but is vulnerable to common attacks like removal or ambiguity attacks (see for example \cite{DBLP:conf/ccnc/SongSML10}).
    \item \textbf{Fine-tuning-based watermarking}: Modifies the diffusion model to introduce watermarks during generation, and a separate trained detector to extract such watermarks \cite{DBLP:conf/iccv/FernandezCJDF23, DBLP:journals/corr/abs-2405-02696}. However, it is possible to remove such watermarks from a model by further fine-tuning them on clean data \cite{hu2024stablesignatureunstableremoving, huang2025watermarkrobustnessradioactivityodds}.
    \item \textbf{Latent space watermarking}: Embeds watermarks in the latent space starting point \cite{DBLP:journals/corr/abs-2410-07369,yang2024gaussianshadingprovableperformancelossless,wen2023treeringwatermarksfingerprintsdiffusion}. Care must be taken when designing such watermarks, as methods like Tree-Ring Watermarking \cite{wen2023treeringwatermarksfingerprintsdiffusion}, alter the latent distribution, making them detectable via steganographic attacks \cite{yang2024can}. Empirical evaluations \cite{DBLP:journals/corr/abs-2410-07369,shehata2024cluemarkwatermarkingdiffusionmodels} also show that Gaussian Shading  introduces artifacts in the generated images \cite{yang2024gaussianshadingprovableperformancelossless} and restricts the diversity of output images \cite{DBLP:journals/corr/abs-2410-07369}. The PRC watermarking scheme \cite{DBLP:journals/corr/abs-2410-07369} is as of now, the only provably undetectable latent space watermarking method, but is only able to encode a limited number ($\sim 500$) of bits robustly if a small amount of white-noise noise is added to the image. 
\end{enumerate}

{%\color{blue}
\subsubsection{Seed-Based Provenance} \label{subsubsection:seed_based}
A recent line of work leverages the coupling between the seed-derived initial noise and the generated output for provenance, NoisePrints is a notable example~\cite{goren2025noiseprintsdistortionfreewatermarksauthorship}.
It derives the initial noise from a cryptographic hash of an author-held
secret. Its verification builds on the empirical observation that a
generated object and its corresponding random starting point are highly correlated
\cite{staniszewski2026there, goren2025noiseprintsdistortionfreewatermarksauthorship}. 
A claim is verified by computing the cosine similarity between the
seed-derived noise and the VAE encoding of the contested object, requiring only
access to the VAE encoder to obtain the object's latent representation.
}

\subsection{Attacks on Watermarking} \label{subsection:watermarking_attack}

Recently, there has been increasing interest in attacks targeting the removal of watermarks from AI Generated Content \cite{DBLP:journals/corr/abs-2108-04974}. Such attacks aim to remove watermarks from a watermarked content while preserving their quality and usability \cite{DBLP:journals/corr/abs-2402-16187}. For instance, \cite{NEURIPS2024_10272bfd} demonstrated that traditional post-processing watermarking schemes are vulnerable to regeneration attacks, where a generative AI model is used to create a perceptually similar but unwatermarked version of the image. Post-processing watermarks are also susceptible to ambiguity attacks \cite{DBLP:conf/ih/LiC04, DBLP:journals/jsac/CraverMYY98}.

Recent research has explored embedding watermarks into the generation process of AIGC by sampling the starting point from a secret watermarked distribution \cite{yang2024gaussianshadingprovableperformancelossless, DBLP:journals/corr/abs-2410-07369, wen2023treeringwatermarksfingerprintsdiffusion}. However, some of these schemes produce outputs which are easily distinguishable with non-watermarked objects \cite{DBLP:journals/corr/abs-2410-07369, shehata2024cluemarkwatermarkingdiffusionmodels} and can be  removed \cite{lee2025removal}. Moreover, empirical evaluations of these schemes show that these watermarks can be removed by adding mild white-noise if the message length is too long (> 1,000 bits) \cite{lee2025removal}. This limits their effectiveness in encoding authorship information in the generated content. 

Forgery of watermarks also remains a significant concern. For example, \cite{yang2024can} proposed a steganographic attack designed to extract watermark patterns, enabling forgery by inserting these patterns into non-watermarked objects. Similar forgery attacks on Tree-Ring and Gaussian Shading leverage their detectable patterns, perturbing the latent space of unwatermarked objects to mimic legitimate watermark signatures \cite{muller2025blackboxforgeryattackssemantic}.

\subsection{Authentication of Digital Content}
Traditionally, digital content (e.g., a document, video, or digital image) is authenticated by digitally signing the content, which is made public to enable verification. To prove authorship, the digital signature should be timestamped using some verifiable timestamping service (see for example \cite{1654294, 10.1145/358790.358798}) -- which allows the first entity who registered the timestamped signature to be designated as the author. 

In settings where the digital content may be altered slightly (e.g. mild noise or compression), this becomes more challenging as small change in the content renders the original signature invalid \cite{10.1145/1671954.1671960}. Efforts to address this include making such signatures robust against mild compression \cite{cryptoeprint:2024/588}, considering robust (perceptual) hashes which tolerate mild distortions \cite{DU2020115713, DBLP:journals/di/McKeownB23}, and constructing hashes robust to certain classes of transformations \cite{7546506}. Nonetheless, some of these approaches suffer from high false reject rates (where two similar content is hashed as not similar) \cite{10.1145/3460120.3484559} or may be circumvented using perceptual-preserving transformations not considered in the hash  -- see \cite{905982} for an example of such transformation.
To date, besides watermarking, we are not aware of  existing work which deals with authentication specifically for AI generated content.

\section{Discussion}

\subsection{Watermarking vs Proof-of-Authorship}
A natural question arises: can we leverage existing watermarking techniques to embed POA information into images generated by LDMs, which can  be extracted to assert ownership during contention (see, for example, \cite{PAGE1998390, DBLP:journals/access/MalanowskaMAMK24})? 

Firstly, existing watermarking schemes typically rely on a secret key for watermark generation and decoding. For example, PRC~\cite{DBLP:journals/corr/abs-2410-07369}, NoisePrints~\cite{goren2025noiseprintsdistortionfreewatermarksauthorship} and Gaussian Shading~\cite{yang2024gaussianshadingprovableperformancelossless} relies on secrets to ensure undetectability.  In contrast, our POA scheme does not rely on any secret. Instead, all POA-related information can be released publicly, enabling transparent and public authorship verification.

Watermarking methods rely on embedding identifiable signals into generated content and subsequently decoding those signals to verify the presence of a watermark. As discussed in Section~\ref{subsection:watermarking_attack}, such approaches are vulnerable to removal and ambiguity attacks, which aim to erase, disrupt, or forge the embedded signal. Consequently, watermarking alone may be insufficient to establish authorship with high confidence.

% Notably, NoisePrints~\cite{goren2025noiseprintsdistortionfreewatermarksauthorship} adopts a parallel approach but is
% fundamentally different from our framework. In NoisePrints, the initial
% random starting point is derived from a cryptographic hash of an
% author-held secret, and this seed-derived initial noise acts as a
% distortion-free watermark. Verification is a watermark-detection test:
% the correlation between the seed-derived initial noise and the contested
% content is checked directly, without regenerating the object or accessing
% the diffusion model. Because the secret must not be revealed during
% disputes, either a trusted verifier service or additional ZKP protocols are needed. 
% In contrast, our POA framework relies on no secret and performs
% authentication rather than detection: the PA evaluates whether the
% contested object corresponds to a claimed identity and generation
% parameters, and outputs a quantified forgery probability, with public
% verification and no cryptographic proof system.

In contrast, our POA framework relies on no secrets, and performs
authentication rather than detection. The authorship information is
cryptographically bound to the generated object. Even with complete
information, it remains prohibitively expensive (due to
Theorem~\ref{thm:poa_PPT}) to derive another set of generation parameters
which maps to the same object under the forger's identity. Furthermore,
our POA framework does not require any changes to the model; authorship
verification remains possible as long as the original model used is
available. Finally, modeling $\mathcal{F}$ as a random oracle
(Section~\ref{subsection:prf}), the starting points generated by our
framework are distributed identically to truly random starting points.
Thus, objects generated using the POA framework are indistinguishable
from images generated normally, and therefore do not
introduce the visible or statistical artifacts observed in some
watermarking schemes.

 \subsection{Practical Considerations} \label{subsection:practical}

\paragraph{Author registration} Similar to digital signature-based authentication systems, our framework
requires the registration of the author's
unique identifier, prior to generation. This registration can be
performed in a timestamped manner, for instance, by recording it in a
trusted ledger or through public, verifiable publication (e.g., on the
author's website, timestamped using a trusted archival service like  Wayback Machine).

\paragraph{Starting Point Distribution}
Most widely-used LDMs (e.g., various versions of Stable Diffusion) sample the initial latent point from $\mathcal{N}(0, I)$. As explained in Section \ref{subsection:gf}, our POA framework does not impose any security assumptions on the starting point sampler $G$. If a model employs a different starting point distribution, $G$ can be replaced accordingly, while deriving the seed in the same manner. Our analysis does not impose any restrictions or assumptions on $G$.

{%\color{blue}
\paragraph{Model Availability During Adjudication} Since the scheme does not rely on secrets, authors should archive the exact model and generation pipeline used. Although not necessary, the authors could choose to publish these information publicly for convenience. 
The PA requires access to these components to reproduce the claimed output and perform verification. For proprietary models, public release may be impractical; instead, verification may be performed by a trusted PA with confidential access (but everything else remains public). If the original model is unavailable, the claim is unverifiable under our framework, rather than invalid.
}
\paragraph{Role of a Public Ledger.}
A public ledger and POA are complementary. A ledger timestamping alone gives proof-of-ownership. POA adds proof-of-authorship, only the true generator can submit a legitimate claim which passes adjudication with non-negligible probability. The public ledger is optional, one may choose to use it for convenience. For example, the ledger then serves as the secure channel for an author to broadcast their identity, and could support transfer of authorship rights. 
The POA framework protects the generated content but not the prompt directly. Extending protection to prompts is an interesting direction, and a possible approach is a timestamped ledger. 

{%\color{blue}
\subsection{Seed-Based Provenance vs POA} \label{subsection:seed_based_vs_poa}

Recall from Section \ref{subsubsection:seed_based} that seed-based schemes like NoisePrints
also anchor provenance in the seed. Nonetheless, NoisePrints differs from POA
on two main axes.

\paragraph{No secret} NoisePrints derives its seed from an author-held secret. Leakage of this secret invalidates the authorship claim, while verification without revealing it requires zero-knowledge proofs. Moreover, competing claims must still be resolved through timestamping, since anyone can generate a secret corresponding to a candidate noise. In contrast, POA does not rely on any secret. Instead, the seed publicly binds authorship information to the generated content through a public cryptographic hash function, enabling public verification via the PA. Furthermore, forging a valid claim requires either a collision in
$\mathcal{F}$ or on the order of $1/p_r$ evaluations
(Theorem~\ref{thm:poa_PPT}), so competing claims are infeasible to produce
in the first place.

\paragraph{Grounding} 
NoisePrints verifies \emph{correlation}: acceptance of a genuine generation relies on an empirically observed relationship between the initial noise and the resulting generation~\cite{staniszewski2026there, goren2025noiseprintsdistortionfreewatermarksauthorship}, without a general theoretical lower bound guaranteeing acceptance of every genuine generation, and with documented failures for low-variance generations. POA instead verifies \emph{reproduction}: from the claimed authorship and generation parameters, the PA deterministically regenerates the claimed original and evaluates its similarity to the contested object. Security against forgery reduces cryptographically to Assumption A1 (Theorem \ref{thm:poa_PPT}), while A2, enables efficient certification of extremely small random-forger probabilities. A2 follows from sufficient, empirically checkable distributional conditions (Section~\ref{sec:assumption}), which are empirically supported in Section~\ref{section:sub-exp-tail}), with a distribution-free fallback when A2 cannot be justified (Section~\ref{sec:distribution-free}).

\paragraph{Adjudication.} The two schemes also differ in how a claim is
evaluated. NoisePrints accepts a claim when the cosine similarity between
the initial starting point (derived from a secret) and the contested object's latent exceeds an
empirically chosen threshold, requiring only the VAE encoder and a single
similarity computation. In our framework, the PA instead estimates the
probability that a random-forger attains the observed similarity score through Monte-carlo sampling,
reports a confidence interval on this probability, and evaluates against
the conservative endpoint of the interval
(Section~\ref{subsection:PA_spec}). The two designs thus occupy
different points on a cost--granularity trade-off. That is, NoisePrints offers
lightweight verification with a binary outcome, while the PA provides
quantified uncertainty at higher computational cost. Moreover, because
the PA derives its sampling randomness deterministically from its inputs
(Section~\ref{subsection:PA_spec}), any party can re-execute the
adjudication on a public claim and reproduce the same verdict, so the
adjudication is publicly checkable and does not require a trusted verifier. 
}

\section{Conclusion}
In this paper, we introduced a proof-of-authorship  framework for latent diffusion models. This framework cryptographically binds an author and generation parameters to the generated object. Unlike prior watermarking approaches, no secrets are needed. Furthermore, our analysis shows that no polynomial-time forger can outperform random guessing when attempting to forge a valid assertion for a generated object.
Finally, empirical evaluations demonstrate that the framework reliably binds authorship information to generated content, while remaining robust against mild-to-moderate perturbations.

\begin{acks}
% This research/project is supported by the National Research Foundation, Singapore under its AI Singapore Programme (AISG Award No: AISG3-RP-2022-029).
We thank the anonymous reviewers for their helpful feedback, which
improved the quality of this manuscript.

\end{acks}

\bibliographystyle{plain}
\bibliography{sample}

\appendix

\section{Estimating $q$} \label{appendix:p_estimate}

As described in Section \ref{subsubsection:PA}, given a pre-specified false-positive rate $\alpha$ and confidence $1-\delta$, the PA estimates a $(1-\delta)$ confidence interval for $q$. Leveraging Assumption A2, together with the mild assumptions that the sub-exponential tail is modeled by a generalized normal distribution (as is common in practice as discussed in Section \ref{subsection:similarity}) and that standard MLE regularity conditions hold,
we show below that $n \in \Omega \left(\ln^2 \frac{1}{\alpha} \ln \frac{1}{\delta}\right)$ independent samples suffices. This result is  stated as Theorem \ref{thm:p_min_samples}. Henceforth, we will set the number of Monte-Carlo samples used by the PA (as described in Section \ref{subsection:PA_spec}) to estimate the similarity distribution to $n = \ln^{2} \frac{1}{\alpha} \ln \frac{1}{\delta}$, which ensures that for a given $\alpha$, the estimation error of $q$ is upper bounded by $\Delta q \leq \frac{q}{5}$ with probability at least $1-\delta$

\begin{theorem} \label{thm:p_min_samples}
    Let $x_1, \dots, x_n$ be independent and identically distributed realizations of random variable $X$, which follows a Generalized Normal Distribution, parameterized by location, shape and scale parameters $\theta = (\mu, \gamma, \beta)$. For a fixed $\alpha \in (0,1)$ and $T \in supp(X)$, let $q(\theta) = \IP[X\geq T]$. To determine if $q(\theta) \leq \alpha$ and  upper bound the estimation error of $q(\theta)$ as $\Delta q \leq \alpha$ with probability at least $1-\delta$, it suffices to use $$n \in \Omega \left(\ln^2 \frac{1}{\alpha} \ln \frac{1}{\delta}\right)$$ samples. 
\end{theorem}

\begin{proof}
Let $\hat{\theta}$ be the MLE approximation of $\theta$, and let $f(x ; \theta)$ and $F(x ; \theta)$ denote the probability and cumulative distribution function, conditional on parameters $\theta$. Without loss of generality, assume that $\mu = 0$.
% For a fixed test statistic $T$, let $p(\theta)$ denote the $p$-value computed using the parameter set $\theta$, i.e. $p(\theta) = 1-F(T;\theta)$. 

Consider the relative error due to estimation error of $\theta$ can be estimated using the Delta method  as
\begin{equation} \label{eqn:delta_p}
    \Delta q \approx \frac{\partial q}{\partial \theta} ||\theta - \hat{\theta}||_2,
\end{equation}
see Corollary 1.1 of \cite{shao2003mathematical} for more details.
It now remains to bound $\frac{\partial q}{\partial \theta}$ and $||\theta - \hat{\theta}||_2$.
To bound $\frac{\partial q}{\partial \theta}$, first observe that 
\begin{equation} \label{eqn:partial_1}
\begin{split}
\frac{\partial q}{\partial \theta}& = - \frac{\partial F(T ; \theta)}{\partial \theta}\\
&= -\frac{\partial}{\partial \theta} \int_{-\infty}^{T} f(x ; \theta)\, dx\\
&= -\frac{\partial}{\partial \theta} \left(1-\int_{T}^{\infty} f(x ; \theta)\, dx\right)\\
&= \frac{\partial}{\partial \theta}\int_{T}^{\infty} f(x ; \theta)\, dx.
\end{split}
\end{equation}
Furthermore, we have
\begin{equation} \label{eqn:partial_2}
    \frac{\partial f(x ; \theta)}{\partial \theta} = f(x ; \theta) \frac{\partial \ln f(x ; \theta)}{\partial \theta},
\end{equation}
Substituting (\ref{eqn:partial_2}) into (\ref{eqn:partial_1}), we have
\begin{equation} \label{eqn:partial_3}
\begin{split}
    \frac{\partial q}{\partial \theta} &= \int_{T}^\infty f(x ; \theta) \frac{\partial \ln f(x ; \theta)}{\partial \theta} \, dx \\
    &= \IE\left[\frac{\partial \ln f(X ; \theta)}{\partial \theta} 1(X \geq T) \right] \\
    &= q  \IE\left[\frac{\partial \ln f(X ; \theta)}{\partial \theta} \middle| X \geq T \right] 
\end{split}
\end{equation}
where expectation is over $X$ which follows a GND distribution parameterized by $\theta$. Since $\ln f(X ; \theta) = C_{\gamma, \beta} - K_{\gamma, \beta} X^\beta$, where $C_{\gamma, \beta}$ and $K_{\gamma, \beta}$ are constants depending on $\gamma$ and $\beta$, it follows that $\frac{\partial \ln f(X ; \theta)}{\partial \theta}  \in O(X^\beta)$. By the tail bound of sub-exponential distributions,  we have $q(\theta) = \IP[X \geq T] \leq \exp( - K_{\gamma, \beta} X^\beta )$, or equivalently, $\ln \frac{1}{q(\theta)} \geq K_{\gamma, \beta} X^\beta $. Thus, 
\begin{equation} \label{eqn:partial_bound_ln}
    \frac{\partial \ln f(X ; \theta)}{\partial \theta} \in O\left( \ln \frac{1}{q(\theta)} \right).
\end{equation}
Substituting (\ref{eqn:partial_bound_ln}) into (\ref{eqn:partial_3}) yields
\begin{equation} \label{eqn:partial_p_ub}
    \frac{\partial q}{\partial \theta} \lesssim q \ln \left(\frac{1}{q} \right).
\end{equation}

Next, we bound $||\theta - \hat{\theta}||_2$, where $\hat{\theta}$ corresponds to the MLE estimate of $\theta$. Under standard MLE regularity conditions, we have $\hat{\theta} \sim \mathcal{N}\left(0, \frac{1}{n} \mathcal{I}(\theta)^{-1} \right)$, where $n$ denotes the number of independent samples used to estimate $\hat{\theta}$, and $\mathcal{I}(\theta)$ is the Fisher information of $\theta$. Equivalently, we have $Z = \sqrt{n}(\hat{\theta} - \theta) \sim \mathcal{N}(0, \Sigma)$, where $\Sigma = \mathcal{I}(\theta)^{-1}$. 

Consider the eigendecomposition $\Sigma = U \Lambda U^T$, and transformation $Y = \Lambda^{-1/2} U^T Z$. Since $\IE[Y] = \Lambda^{-1/2} U^T 0 = 0$, and $Cov(Y) =\Lambda^{-1/2} U^T \Sigma (\Lambda^{-1/2} U^T)^T = I$, we have that $Y \sim \mathcal{N}(0, I)$. Since $\dim(\theta) =3$, it follows that $Y^T Y  \sim \chi^2_3$. Next, we observe that $Z^T Z = Y^T \Lambda Y \leq \lambda_{\max} Y^T Y$, where $\lambda_{\max}$ is the largest eigenvalue of $\mathcal{I}(\theta)^{-1}$. Therefore, we have $||\hat{\theta} - \theta||_2^2 \leq \frac{\lambda_{\max}}{n} Y^T Y$. Thus, with probability at least $1-\delta$, it follows that 
\begin{equation} \label{eqn:theta_bound_delta}
||\hat{\theta} - \theta||_2 \leq \sqrt{\frac{\lambda_{\max}}{n} \chi^2_{3, 1-\delta}} \lesssim \sqrt{\frac{\lambda_{\max}}{n} \left(3 + \ln \frac{1}{\delta}\right)}
\end{equation}
where $ \chi^2_{3, 1-\delta}$ denotes the $1-\delta$ quantile of the $\chi^2_3$ distribution, and the rightmost bound is due to the Laurent-Massart bound (see Lemma 1 of \cite{laurent}). 

Substituting (\ref{eqn:partial_p_ub}) and (\ref{eqn:theta_bound_delta}) into (\ref{eqn:delta_p}), we obtain 
\[
    \Delta q \lesssim q \ln \left(\frac{1}{q} \right)\sqrt{\frac{\lambda_{\max}}{n} \left(3 + \ln \frac{1}{\delta}\right)}
\]

Finally, by Assumption A2, the similarity scores follow a sub-exponential distribution. This implies finite second moments and, consequently, a bounded second moment of the score function. As a result, the Fisher information matrix has bounded operator norm, and in particular its largest eigenvalue $\lambda_{\max}$ is finite. Thus, we have,
\[
\Delta q \in O\left( q \ln\frac{1}{q} \sqrt{\frac{1}{n} \ln \frac{1}{\delta}} \right).
\]
To bound $\Delta q \leq \alpha$, it suffices for $$n = \Omega \left(\ln^2 \frac{1}{\alpha} \ln \frac{1}{\delta}\right)$$. 
\end{proof}

\section{Proof of Theorem \ref{thm:poa_PPT}, Section \ref{subsubsection:impossibility}} \label{appendix:poa_PPT}

Let $\mathcal{I}$ denote the set of all valid identities,  $\mathcal{F} = \{f_i : i \in \mathcal{I} \}$ denote the corresponding PRF family indexed over $\mathcal{I}$, and $\mathcal{U} = \{0,1\}^{n_{seed}}$ denote all possible $n_{seed}$-bit binary sequences (i.e. all possible outputs of $f_i$). The adversary has full knowledge of the function family $\mathcal{F}$ and a query budget $Q$ (as defined in Game \ref{game:forge}). 
To formalize the impossibility of a polynomial-time forger over a random-forger with non-negligible advantage, consider the following forgery game.

\begin{game}[Seed-Collision Game] \label{game:collision}
The challenger is parameterized by arbitrary $i \in \mathcal{I}$, $m$, $e$.
\begin{enumerate}
    \item The challenger uniformly samples $r$ and sends $(i, m, e, r)$ to
    the adversary. The adversary may evaluate $f_k(\cdot)$ for any
    $k \in \mathcal{I}$ at will, as identities are public.
    \item The adversary outputs $(j, \widetilde{m}, \widetilde{e}, \widetilde{r})$ with
    $j \in \mathcal{I}$, and wins if and only if
    $(j, \langle\widetilde{m}, \widetilde{e}, \widetilde{r}\rangle) \neq
    (i, \langle m, e, r\rangle)$ and
    $f_j(\langle \widetilde{m}, \widetilde{e}, \widetilde{r} \rangle) =
    f_i(\langle m, e, r\rangle)$.
\end{enumerate}
\end{game}
Note that in Game \ref{game:collision}, since $\mathcal{F}$ is a deterministic PRF family, the adversary is free to compute $f_i(\cdot) \in \mathcal{F}$ for any $i \in \mathcal{I}$. 

\begin{lemma} \label{lemma:no_collision}
If $\mathcal{F}$ satisfies Assumption A1, then for all $i \in \mathcal{I}$,
$m$, $e$, no PPT adversary wins Game~\ref{game:collision} with
non-negligible probability.
\end{lemma}

\begin{proof}
Suppose there exists a $i, m, e$ such that a PPT adversary $\mathcal{A}$ wins Game~\ref{game:collision} with
non-negligible probability. We construct a PPT adversary $\mathcal{B}$ that finds
a collision over $\mathcal{F}$:
$\mathcal{B}$ simulates the challenger by uniformly sampling $r$ and sending
$(i, m, e, r)$ to $\mathcal{A}$; upon receiving
$(j, \widetilde{m}, \widetilde{e}, \widetilde{r})$, it outputs the pair
$\big((i, \langle m,e,r\rangle),\, (j, \langle \widetilde{m}, \widetilde{e},
\widetilde{r}\rangle)\big)$. If $\mathcal{A}$ is PPT, so is $\mathcal{B}$.
Whenever $\mathcal{A}$ wins, the two key--input pairs are distinct and map
to the same output, so $\mathcal{B}$ outputs a valid collision, contradicting
the unkeyed collision resistance of $\mathcal{F}$
(Assumption~A1). Thus no such adversary exists.
\end{proof}

\begin{game}[Seed-Steering Game] \label{game:steering}
The challenger is parameterized by the target rate $p_r$
\begin{enumerate}
    \item The adversary chooses an identity $i \in \mathcal{I}$,
    meta-parameters $m$, and an embedding $e$, and sends $(i, m, e)$
    to the challenger. The adversary may evaluate $f_k(\cdot)$ for any
    $k \in \mathcal{I}$ subject to the query budget $Q$.
    \item The challenger uniformly samples $r$, sets
    $\kappa = \langle m, e, r\rangle$ and
    $I = L_m\big(e,\, G(f_i(\kappa))\big)$, and sends $(I, r)$ to the
    adversary.
\item The adversary outputs $(j, \widetilde{m}, \widetilde{e}, \widetilde{r}
)$ with $j \in \mathcal{I}$,
inducing the seed $\widetilde{u} = f_j(\langle \widetilde{m}, \widetilde{e},
\widetilde{r}\rangle)$ and $\widetilde{T} =
\mathit{Sim}\big(L_{\widetilde{m}}(\widetilde{e}, G(\widetilde{u})),
I\big)$. The adversary wins iff
(i) $j \neq i$, (ii) $\widetilde{u} \neq f_i(\kappa)$, and (iii)
\[
\Pr_{u \leftarrow \mathcal{U}}\Big[
\mathit{Sim}\big(L_{\widetilde{m}}(\widetilde{e}, G(u)),\,
I\big) \ge \widetilde{T}
\Big] \le p_r.
\]
where $u \leftarrow \mathcal{U}$ denotes sampling uniformly from all possible seeds. 
\end{enumerate}
\end{game}

\begin{lemma} \label{lemma:correlation_intractable}
If $\mathcal{F}$ satisfies Assumption A1, then every PPT adversary with query
budget $Q$ wins Game~\ref{game:steering} with probability at most
$Q \cdot p_r + \mathrm{negl}(\lambda)$.
\end{lemma}
\begin{proof}
Let $\kappa = \langle m, e, r\rangle$ and $I$ be as defined in Steps
(1) and (2) of Game~\ref{game:steering}. Consider the relation
\[
\begin{split}
R_{I,i} = \Big\{ \big((j, \widetilde{\kappa}),\, u\big) \;:\;
j \neq i \;\wedge\; \Pr_{u' \leftarrow \mathcal{U}}\Big[&
\mathit{Sim}\big(L_{\widetilde{m}}(\widetilde{e}, G(u')),\, I\big) \\[-2pt]
&\ge \mathit{Sim}\big(L_{\widetilde{m}}(\widetilde{e}, G(u)),\, I\big)
\Big] \le p_r \Big\},
\end{split}
\]
where $p_r$ is as defined in Game~\ref{game:steering} and
$\widetilde{m}, \widetilde{e}$ are parsed from $\widetilde{\kappa}$. In other
words, $R_{I,i}$ pairs candidate claims with the seeds $u$ for which
the ideal adjudicator would accept: any $((j,\widetilde{\kappa}), u) \in
R_{I,i}$ corresponds to a valid forgery claim for the image $I$
generated using $(i, \kappa)$, provided $f_j(\widetilde{\kappa}) = u$.

$R_{I,i}$ is sparse with density at most $p_r$: for every fixed
$(j, \widetilde{\kappa})$, membership of $((j,\widetilde{\kappa}), u)$
requires the score of $u$ to lie in the upper $p_r$-tail of its own
distribution over uniformly random seeds, which has probability at most
$p_r$ by definition of the tail, since $p_r$ is defined as a tail mass
under $G(u)$ for uniform $u$.

Given an adversary $\phi$ with query budget $Q$ in
Game~\ref{game:steering}, we construct an adversary $\Gamma$ against
the correlation intractability of $\mathcal{F}$ for $R_{I,i}$ by
simulating the challenger for $\phi$:
\begin{enumerate}[label=\alph*.]
\item In Step 1, $\Gamma$ receives $(i, m, e)$ from $\phi$.
\item In Step 2, $\Gamma$ samples $r$ uniformly, sets
$\kappa = \langle m, e, r\rangle$, computes $u = f_i(\kappa)$, and
sends $(I, r)$ to $\phi$, where $I = L_m(e, G(u))$.
\item In Step 3, $\Gamma$ receives $(j, \widetilde{m}, \widetilde{e},
\widetilde{r})$ from $\phi$.
\item $\Gamma$ outputs $\big((j, \widetilde{\kappa}),
f_j(\widetilde{\kappa})\big)$, where $\widetilde{\kappa} = \langle \widetilde{m},
\widetilde{e}, \widetilde{r}\rangle$.
\end{enumerate}
This is a perfect simulation of Game~\ref{game:steering}, $\Gamma$ makes
the same $Q$ evaluations of $\mathcal{F}$ as $\phi$ (plus one for its
output, which we charge to $Q$), and $\phi$ wins exactly when
$\Gamma$'s output lies in $R_{I,i}$. By the correlation intractability
of $\mathcal{F}$ for $R_{I,i}$ assumed in A1, an adversary making $Q$ evaluations finds a
member of a relation of density at most $p_r$ with probability at most
$Q \cdot p_r + \mathrm{negl}(\lambda)$. Hence $\phi$ wins
Game~\ref{game:steering} with probability at most
$Q \cdot p_r + \mathrm{negl}(\lambda)$.
\end{proof}
\begin{remark}[Concrete bound and the role of $p_r$] \label{remark:concrete_sparsity}
Lemma~\ref{lemma:correlation_intractable} bounds the success probability
of an adversary with query budget $Q$ by $Q \cdot p_r + \mathrm{negl}(\lambda)$.
The threshold $p_r$ is therefore a per-evaluation rate rather than a total
success probability: in the deployed configuration ($p_r = 2^{-50}$,
$\lambda = 256$), an adversary making $Q$ evaluations of $\mathcal{F}$
succeeds with probability at most $Q \cdot 2^{-50} + \mathrm{negl}(\lambda)$,
where the negligible term is the error in the correlation-intractability
bound of Assumption~A1. For example, $Q = 2^{30}$
evaluations yield a bound of approximately $2^{-20}$.
\end{remark}

\begin{remark}[Handling the transformation $\widetilde{t}$]
\label{remark:transformation}
For clarity, the proof of Lemma~\ref{lemma:correlation_intractable}
suppresses the transformation $\widetilde{t}$ in the adversary's claim. Its
inclusion does not affect the argument: for each fixed
$\widetilde{t} \in \mathcal{T}$, the acceptance threshold is by definition
the $(1-p_r)$-quantile of the score distribution conditional on
$(\widetilde{m}, \widetilde{e}, \widetilde{t})$, so the corresponding relation
$R_{I,i,\widetilde{t}}$ (defined as in the proof with $I$ replaced by
$\widetilde{t}(\mathcal{E}(I))$) has density at most $p_r$ by the identical
tail-mass argument. An adversary choosing $\widetilde{t}$ adaptively
selects a member of the relation family
$\{R_{I,i,\widetilde{t}}\}_{\widetilde{t}\in\mathcal{T}}$; correlation
intractability applies to each member, and the reduction of
Lemma~\ref{lemma:correlation_intractable} carries over with $\Gamma$
forwarding $\phi$'s chosen $\widetilde{t}$ as the relation index.
Concretely, an adversary attempting $q_t$ distinct transformations
gains at most a factor $q_t$ in the union bound, so the bound of
Theorem~\ref{thm:poa_PPT} becomes $Q \cdot q_t \cdot p_r + \mathrm{negl}(\lambda)$
unless transformation attempts are charged to $Q$; the judge's
screening of unreasonable transformations
(Section~\ref{subsubsection:PA}) further bounds the effective
$|\mathcal{T}|$.
\end{remark}
% To formalize the impossibility of a polynomial-time forger outperforming a
% random-forger, we model $f$ as a keyed function randomly sampled from $\mathcal{F}$. accessible only by query,
% consistent with Assumption A1 (Section~\ref{sec:assumption}): on a
% key--input pair not previously queried, the response is indistinguishable
% from a uniformly sampled $n_{\mathrm{seed}}$-bit sequence. Consider the following PRF indistinguishability game:
% \begin{game}[PRF Indistinguishability Game]
% \label{game:prf}
% The challenger and adversary interact as follows:
% \begin{enumerate}
%     \item The challenger samples $\sigma \leftarrow \mathrm{Unif}(\{0,1\})$.
%     If $\sigma = 1$, it sets $\mathcal{O} = f$; otherwise it sets
%     $\mathcal{O} = f^{\mathrm{rand}}$, a truly random function over
%     key--input pairs.
%     \item The adversary is allowed a polynomial number of queries to
%     $\mathcal{O}$ on key--input pairs of its choice.
%     \item The adversary outputs $\hat{\sigma}$ and wins if and only if
%     $\hat{\sigma} = \sigma$.
% \end{enumerate}
% \end{game}

% By Assumption A1, no PPT adversary wins Game~\ref{game:prf} with
% non-negligible advantage over random guessing. If such an adversary exists, the adversary could then distinguish a PRF from a truly random function with non-negligible success rate. 

\begin{lemma}[Advantage in Game~\ref{game:forge} Implies Advantage in
Game~\ref{game:collision} or Game~\ref{game:steering}]
\label{lem:forgery_implies_cr_ci}
Let $\mathcal{A}$ be a PPT adversary with query budget $Q$ in
Game~\ref{game:forge}. If
$\mathbf{Adv}^{\mathsf{forge}}_{\mathcal{A},i}(\lambda, Q) - Q \cdot p_r$
is non-negligible, then there exists a PPT adversary that either wins
Game~\ref{game:collision} with non-negligible probability, or wins
Game~\ref{game:steering} with query budget $Q$ and probability exceeding
$Q \cdot p_r$ by a non-negligible amount.
\end{lemma}
\begin{proof}
Suppose we are given an object $I$ generated under author identity
$i$ with parameters $\kappa = \langle m, e, r\rangle$, and let
$u^\ast = f_i(\kappa)$. Since the forger's claim
$(j, \widetilde{\kappa})$ induces the seed $\widetilde{u} =
f_j(\widetilde{\kappa})$, which either equals $u^\ast$ or does not, any
winning forger satisfies exactly one of the following:
\begin{enumerate}
\item \textbf{Collision.} Output $(j, \widetilde{\kappa})$ with $j \neq
i$ and $f_j(\widetilde{\kappa}) = u^\ast$. That is, the forger induces a
collision such that the LDM samples the same starting point.
\item \textbf{Steering.} Output $(j, \widetilde{\kappa})$ with $j \neq i$
and $\widetilde{u} \neq u^\ast$, whose induced seed satisfies
$\big((j, \widetilde{\kappa}), \widetilde{u}\big) \in R_{I,i}$, where
$R_{I,i}$ is the relation defined in the proof of
Lemma~\ref{lemma:correlation_intractable}; membership in $R_{I,i}$ is
exactly condition (ii) of Game~\ref{game:forge}.
\end{enumerate}

These two cases are exhaustive and mutually exclusive. Every claim
$(j,\widetilde{\kappa})$ induces a seed $\widetilde{u} = f_j(\widetilde{\kappa})$,
which either equals $u^\ast$ or does not; each execution of the game
therefore falls into exactly one case, according to this realized
bit. The partition is over executions rather than over adversaries
or their strategies, and hence, covers every possible adversary. Note
that the two cases here map to Games~\ref{game:collision}
and~\ref{game:steering}, respectively: an adversary that runs
$\mathcal{A}$ and forwards its claim wins Game~\ref{game:collision} in
the Collision case and Game~\ref{game:steering} (with the same query
budget $Q$) in the Steering case. Therefore
$\mathbf{Adv}^{\mathsf{forge}}_{\mathcal{A},i}(\lambda, Q)$ is at most the
sum of the two winning probabilities, and if it exceeds $Q \cdot p_r$
by a non-negligible amount then so does one of the summands (the first
relative to $0$, the second relative to $Q \cdot p_r$), since a sum of
two negligible functions is negligible.
\end{proof}

\begin{theorem*}[Theorem~\ref{thm:poa_PPT}, Section~\ref{sec:securityanalysis}]
Under Assumption A1, for every identity $i \in \mathcal{I}$, query budget
$Q$ and random-forger success rate $p_r$ (as defined in
Game~\ref{game:forge}), and every PPT adversary $\mathcal{A}$ with query
budget $Q$,
\[
\mathbf{Adv}^{\mathsf{forge}}_{\mathcal{A},i}(\lambda, Q)
\;\le\; \mathrm{negl}(\lambda) + Q \cdot p_r,
\]
where $\lambda$ is the security parameter of the family $\mathcal{F}$.
\end{theorem*}

\begin{proof}
We prove this claim by reduction. Suppose there exists a PPT adversary
$\mathcal{A}$ with query budget $Q$ such that
$\mathbf{Adv}^{\mathsf{forge}}_{\mathcal{A},i}(\lambda, Q) - Q \cdot p_r$
is non-negligible.

By Lemma~\ref{lem:forgery_implies_cr_ci}, $\mathcal{A}$ can be used to
construct an adversary that either wins Game~\ref{game:collision} with
non-negligible probability, or wins Game~\ref{game:steering} with query
budget $Q$ and probability exceeding $Q \cdot p_r$ by a non-negligible
amount. The former contradicts Lemma~\ref{lemma:no_collision}, which
bounds the success probability in Game~\ref{game:collision} by
$\mathrm{negl}(\lambda)$; the latter contradicts
Lemma~\ref{lemma:correlation_intractable}, which bounds the success
probability in Game~\ref{game:steering} by
$Q \cdot p_r + \mathrm{negl}(\lambda)$. Hence
$\mathbf{Adv}^{\mathsf{forge}}_{\mathcal{A},i}(\lambda, Q) - Q \cdot p_r$
is negligible, and the claim follows.
\end{proof}

\section{Proofs in Section \ref{subsection:similarity}} \label{appendix:proofs_similarity}

\begin{lemma*}[\emph{Lemma \ref{lemma:subexp}, Section \ref{subsection:similarity}}]

Let
\[
T
:=
\frac{1}{d}\langle \mathcal{L}, \mathcal{L}'\rangle
=
\frac{1}{d}\sum_{j=1}^{d}\mathcal{L}_j \mathcal{L}'_j.
\]
Under Assumption~\textnormal{(A2-I)}, there exists constants $c, K>0$
% there exist universal constants
% \(c,C>0\) such that
% \[
% \left\|
% T-\mathbb{E}[T]
% \right\|_{\psi_1}
% \leq CK^2,
% \]
% where $\| \cdot \|_{\psi_p}$ denotes the Orlicz norm for $\psi_p(x) = \exp(x^p)-1$.
% Consequently, 
for every \(t\geq 0\),
\[
\mathbb{P}\left(
\left|T-\mathbb{E}[T]\right|
\geq t
\right)
\leq
2\exp\left(
-c\frac{t}{K^2}
\right) \in \exp(-\Omega(t)),
\]
i.e. the tail of the probability distribution is upper bounded by a sub-exponential tail. 
\end{lemma*}

\begin{proof}
Let
\[
Z_j
:=
X_jY_j-\mathbb{E}[X_jY_j].
\]
By Assumption~\textnormal{(A2-I)}, since $X_j$ and $Y_j$ are sub-Gaussian there exists \(K>0\) such that,
for every \(j\),
$
\|X_j\|_{\psi_2},
\|Y_j\|_{\psi_2}\le K;
$
see \citep[Section~2.6.1]{Vershynin2026HDP}.

By \citep[Lemma~2.8.6]{Vershynin2026HDP}, the product of two
sub-Gaussian random variables is sub-exponential and satisfies
\[
\|X_jY_j\|_{\psi_1}
\le
\|X_j\|_{\psi_2}\|Y_j\|_{\psi_2}
\le K^2.
\]
Using the sub-exponential centering inequality
\citep[equation~(2.26) and Exercise~2.44(a)]
{Vershynin2026HDP}, we obtain
\[
\|Z_j\|_{\psi_1}
\le
C\|X_jY_j\|_{\psi_1}
\le CK^2,
\]
where \(C>0\) is an absolute constant.

Since
\[
T-\mathbb{E}[T]
=
\frac{1}{d}\sum_{j=1}^d Z_j,
\]
and \(\|\cdot\|_{\psi_1}\) is a norm
\citep[Exercise~2.42]{Vershynin2026HDP}, the triangle inequality yields
\[
\begin{aligned}
\|T-\mathbb{E}T\|_{\psi_1}
&=
\left\|
\frac{1}{d}\sum_{j=1}^d Z_j
\right\|_{\psi_1} \\
&\le
\frac{1}{d}\sum_{j=1}^d
\|Z_j\|_{\psi_1} \\
&\le CK^2.
\end{aligned}
\]

Let
$
L:=2CK^2.
$
Since \(L>\|T-\mathbb{E}T\|_{\psi_1}\), the definition of the
sub-exponential Orlicz norm implies
\[
\mathbb{E}\exp\left(
\frac{|T-\mathbb{E}T|}{L}
\right)
\le 2.
\]
Therefore, by Markov's inequality,
\[
\begin{aligned}
\mathbb{P}\left(
|T-\mathbb{E}T|>t
\right)
&=
\mathbb{P}\left[
\exp\left(
\frac{|T-\mathbb{E}T|}{L}
\right)
>
\exp\left(\frac{t}{L}\right)
\right] \\
&\le
\exp\left(-\frac{t}{L}\right)
\mathbb{E}\exp\left(
\frac{|T-\mathbb{E}T|}{L}
\right) \\
&\le
2\exp\left(-\frac{t}{2CK^2}\right).
\end{aligned}
\]
Hence, for some absolute constant \(c>0\),
\[
\mathbb{P}\left(
|T-\mathbb{E}T|>t
\right)
\le
2\exp\left(
-c\frac{t}{K^2} 
\right) \in \exp(-\Omega(t)),
\]
and the sub-exponential tail follows. 
\end{proof}

\begin{theorem*} [\emph{Theorem \ref{theorem:hypothesis}, Section \ref{subsection:similarity}}]
   Let $Sim(Z_1, Z_2) = \frac{1}{d}Z_1 \cdot Z_2$ denote the similarity score between the latent vectors $Z_1$ and  $Z_2$. For tractability, assume that $Sim(Z_1, Z_2) $ follows a generalized normal distribution with mean $\mu$, and common scale and shape parameters   $\gamma$ and $\beta$. Consider the hypothesis test $H_0: \mu = 0$ (dissimilar latents) v.s. $H_1: \mu > 0$ (similar latents), and the test statistic $T(Z_1, Z_2) = 1(\max(0,Sim(Z_1, Z_2)) \geq W_\alpha^I)$, where $W_I^\alpha$ is the $(1-\alpha)$ quantile of $W$ (as defined in Section \ref{subsection:PA_spec}), given $I$. It follows that $T$ is a uniformly most powerful test among all tests with a false-positive rate of at most $\alpha$.
\end{theorem*}
\begin{proof}
    Under the assumption that $q = \max(0, Sim(Z_1, Z_2))$ follows a generalized normal distribution, and admits pdf $f(x; \mu, \gamma, \beta) \propto \exp \left( - \left( \frac{|x-\mu|}{\gamma}\right)^\beta\right)$. 
    Consider the following likelihood ratio computed over one pair $(Z_1, Z_2)$, 
    \[
    \begin{split}
    l(q) &= \frac{\sup_{\mu \in H_1 \cup H_0} f(q : \mu, \gamma, \beta)}{\sup_{\mu \in H_0} f(q : \mu, \gamma, \beta)} \\
    &= \frac{f(q : q, \gamma, \beta)}{f(q :0, \gamma, \beta)} \\
    &= \exp \left(  \left( \frac{q}{\gamma}\right)^\beta\right),
    \end{split}
    \]
    which is monotonically increasing in $q$. 
    Thus, by the Neymann-Pearson lemma, the test $T(Z_1, Z_2) = 1(q > W_\alpha^I)$ is the uniformly most powerful test of power $\alpha$ for $H_0$ against $H_1$. 
\end{proof}

\section{A Sharper Bernstein-Type Inequality for Tail Bounds} \label{app:stronger-tail}

Recall from Remark~\ref{remark:subexp_bernstein} that a sharper tail
bound is attainable under an additional weak-dependence assumption on
the latent coordinates. Like (A2-I), this assumption is a mild
relaxation of the idealized VAE prior: under the ideal
$\mathcal{N}(0, I)$ prior (Section~2.1.1) the latent coordinates are
mutually independent, and weak dependence relaxes exact independence
to the requirement that coordinates far apart be nearly independent.
Consider a $d$-dimensional VAE latent $X = (X_1, \dots, X_d)$, we formalize the
assumption as follows.

\paragraph{Assumption (A2-II): weak dependence.}
The latent coordinate sequence $(X_j)$ is weakly dependent, in the
sense that its strong ($\alpha$-) mixing coefficients decay
geometrically: $\alpha(n) \le \exp(-2cn)$ for some constant $c > 0$,
i.e., condition (1.3) of \cite{merlevede2009bernstein}. Informally, coordinates far
apart in the enumeration are nearly independent: the mixing
coefficient $\alpha(n)$ measures the maximal dependence between the
``past'' $\sigma(X_j : j \le k)$ and the ``future''
$\sigma(X_j : j \ge k+n)$ over all events --- not merely
correlations --- and (A2-II) requires this dependence to decay
exponentially with the index gap $n$.
 
Weak-dependence assumptions of this form are standard in
high-dimensional statistics and time-series analysis \cite{Vershynin2026HDP}. The diagnostics
of Section~\ref{section:sub-exp-tail} are consistent with (A2-II) under the above ordering:
correlations above the sampling noise floor occur only between
same-site channels and spatially adjacent sites, and vanish at large
index separations --- the profile of finite-range dependence, for
which the mixing coefficients vanish beyond the dependence range and
(A2-II) holds automatically. Since such diagnostics support but cannot
conclusively establish a mixing condition, we impose (A2-II) as an
assumption; all formal guarantees of the framework rely only on
Lemma \ref{lemma:subexp}.
 
\paragraph{Sharper tail bound.}
Suppose additionally that the latent coordinates are uniformly bounded
by a constant $M$, which likely holds for most VAE implementations, or could be clipped manually using a threshold determined empirically. Since $\mathcal{L} \perp \mathcal{L}'$, the centered
products $Z_j = X_j Y_j - \mathbb{E}[X_j Y_j]$ are bounded by $2M^2$
and inherit the geometric mixing rate. Applying Theorem 1 of
\cite{merlevede2009bernstein}
(which requires no stationarity) to $\sum_{j=1}^d Z_j$ and rescaling
by $1/d$, the similarity statistic
$T = \frac{1}{d}\langle \mathcal{L}, \mathcal{L}'\rangle$ satisfies,
for constants depending only on $c$ and $M$ for sufficiently large $t$,
\[
\mathbb{P}\big(|T - \mathbb{E}T| \ge t\big)
\;\le\;
\exp\!\left(- \Omega\!\left(\frac{d\, t}{(\log d)(\log\log d)}\right)\right) = \exp(-\widetilde{\Omega}(d)t).
\]
This
sharpens the dimension-free exponent $t$ of Lemma \ref{lemma:subexp} by a
factor of $d$ up to logarithmic corrections. These corrections are
intrinsic to the dependent setting, as the large deviations principle
fails for geometrically mixing sequences (see the discussion in
\cite{merlevede2009bernstein}). Nonetheless, we emphasize that this sharpened bound is a supplementary result: it is not needed for Assumption (A2), whose sub-exponential tail already follows from (A2-I) alone via Lemma \ref{lemma:subexp}, without any dependence condition.

\section{$\ell_p$ Adversarial Perturbations on Similarity Function} \label{appendix:adversarial}
We  discuss the performance of $\ell_p$ bounded adversaries, for general $p\geq 1$. The notations will be reused from Section \ref{subsection:adv_attack}. Recall that for a given $p$, the adversary constructs $\widetilde{\mathcal{L}} = \mathcal{L} + v$, subject to the constraint $||v||_p \leq \epsilon$, with the goal of minimizing $Sim(\mathcal{L}, \widetilde{\mathcal{L}}) = \frac{1}{d} (\mathcal{L} \cdot \mathcal{L} + \mathcal{L} \cdot v)$. 

For a given $p\geq 1$, define its Holder conjugate $q = \frac{p}{p-1}$. By dual-norm of $\ell_p$ spaces, it can be shown that $|\mathcal{L} \cdot v| \leq \epsilon ||\mathcal{L}||_q$. 

Denote $v_i$ as the $i$-th coordinate of $v$. It can be shown that
\[
v_i = - \epsilon \, \frac{sign(\mathcal{L}_i) \, |\mathcal{L}_i|^{\,q-1}}{\|\mathcal{L}\|_q^{\,q-1}},
\]
for $i=1, \dots, d$ attains the minimum value $\mathcal{L}\cdot v = - \epsilon ||\mathcal{L}||_q$, and minimizes $Sim(\mathcal{L}, \widetilde{\mathcal{L}}) = Sim(\mathcal{L}, \mathcal{L}) - \frac{1}{d}\epsilon ||\mathcal{L}||_q$. Thus, at $p=1$, $Sim(\mathcal{L}, \widetilde{\mathcal{L}}) =  Sim(\mathcal{L}, \mathcal{L}) - \frac{1}{d}\epsilon ||\mathcal{L}||_\infty$, and as $p$ increases, $Sim(\mathcal{L}, \widetilde{\mathcal{L}})  \rightarrow Sim(\mathcal{L}, \mathcal{L}) - \frac{1}{d}\epsilon ||\mathcal{L}||_1$. Furthermore, it is also straightforward to show that $||\mathcal{L}||_q \in o(d)$. Thus for reasonable values of $\epsilon$ which preserves similarity between $\mathcal{L}$ and $\widetilde{\mathcal{L}}$, the impact of 
$\frac{1}{d}\epsilon ||\mathcal{L}||_q$ on the overall similarity is negligible.

\subsection{Distance Preservation Between Starting Points and Latents} \label{section:expt_distance}

\subsubsection{Seed-Induced Separation of Generated Objects} \label{subsection:seed_induced}

% In our POA framework, the seed is constructed by passing the author's unique identifier and
% generation parameters into a PRF $f_i$. 
% The seed is then passed into a pseudorandom number generator $G$ to generate an initial random starting point
% for the diffusion process according to some distribution (e.g. multivariate Gaussian distribution).  
% We can show that two distinct seeds used to generate a random starting point from a 
% standard multivariate Gaussian distribution, then with high probability,
% normalized squared $\ell_2$
%  distance\footnote{Normalized squared $\ell_2$ distance between two vectors of equal lengths $X$ and $Y$ is ${\frac{1}{|X|} \sum_{\gamma \in X-Y} \gamma^2}$} between the two starting points is in the range
% $2 \pm 3 \times \sqrt{\frac{8}{d}}$, where $d$ is the number of dimensions of the starting point (in Stable Diffusion, 
% $d = 4 \times 64 \times 64$).

% \paragraph{Relationship Between Initial Starting Point and Generated Content}
We first establish that independently sampled starting points remain
separated after the diffusion process. This provides the basic
geometric justification for using the generated latent to distinguish
the claimed object from objects produced using unrelated seeds. We  present this result in Lemma \ref{lemma:random_latent_rs}

\begin{lemma} \label{lemma:random_latent_rs}
    Let $s, s'$ denote two initial starting points sampled independently from $\mathcal{N}(0, I)$, and $\mathcal{L}, \mathcal{L}'$ denote the resulting latents after performing DDIM denoising on $s, s'$ over $t$ timesteps. The following relationship holds:
    \[
    \IE[||\mathcal{L} - \mathcal{L}'||_2] \geq f(\alpha_t) \IE[||s - s'||_2],
    \]
    where  $f(\alpha_t) = \frac{1 - \sqrt{1-\alpha_t}}{ \sqrt{\alpha_t}}$ and $\alpha_t>0$ is a small constant from the DDIM scheduler which is decreasing in $t$. 
\end{lemma}

\begin{proof}
If the DDIM scheduler is used, the relationship between $\mathcal{L}$ and $s$ is given in \cite{DBLP:conf/iclr/SongME21} (Equation (4)):
    \begin{equation} \label{eqn:ddim}
       s   = \sqrt{\alpha_t} \mathcal{L}+ \sqrt{1-\alpha_t} \epsilon,
    \end{equation}
where $\alpha_t\in (0,1)$ is a scheduler-specific constant decreasing in $t$. 
% For example, $\alpha_1 \approx 0.999$ and $\alpha_{100} \approx 0.894$ in DDIM. 
    From (\ref{eqn:ddim}), 
    \[
       s - s'  = \sqrt{\alpha_t} (\mathcal{L} - \mathcal{L}')+ \sqrt{1-\alpha_t} (\epsilon - \epsilon') ,
    \]
    where $\epsilon' = \frac{s' - \sqrt{\alpha_t} \mathcal{L}}{\sqrt{1-\alpha_t}}$ (from  (\ref{eqn:ddim})). 
    From the triangle inequality, 
    \begin{equation} \label{eqn:bound1}
        ||s - s'||_2  \leq \sqrt{\alpha_t} ||\mathcal{L} - \mathcal{L}'||_2+ \sqrt{1-\alpha_t} ||\epsilon - \epsilon'||_2. 
    \end{equation}
    Since both $s - s'$ and $\epsilon - \epsilon'$ are $\mathcal{N}(0, 2I)$ random variables, $\IE[ ||s - s'||_2 ] = \IE[||\epsilon - \epsilon'||_2]$. Therefore, taking the expectation on (\ref{eqn:bound1}) and rearranging the terms, the bound follows. 

 \end{proof}

In particular, Lemma \ref{lemma:random_latent_rs} tells us that 
given a pair of independently generated starting points $(s, s')$ and their corresponding generated objects $( \mathcal{L},  \mathcal{L}')$, the expected value of $||\mathcal{L} - \mathcal{L}'||_2$ is lower bounded by a constant factor of $||s - s'||_2$. That is, if $s$ and $ s'$ are far apart, then $ \mathcal{L}$  and $\mathcal{L}'$ will also be far apart. 

\subsubsection{Empirical Validation}
We empirically validate Lemma~\ref{lemma:random_latent_rs} by showing that the distance between generated latents $\mathcal{L}$ and $\mathcal{L}'$ admits a constant-factor lower bound in terms of the distance between their respective starting points $s$ and $s'$. More specifically, given a pair of text embedding $e$ and LDM meta-parameter $m$, 
\begin{enumerate}
    \item Randomly sample 30 starting points $s_1, \cdots, s_{30}$ independently from $\mathcal{N}(0, I)$
    \item Generate the 30 corresponding latents $\mathcal{L}^1, \dots, \mathcal{L}^{30}$, where $\mathcal{L}^i = L_m(e, s_i)$ denotes the latent generated using the LDM conditional on $e$, $m$ and starting point $s_i$.
    \item Compute all ${30 \choose 2}$ pairwise ratios of $\ell_2$ distances $\frac{||\mathcal{L}^i - \mathcal{L}^j||_2}{||s_i - s_j||_2}$ for $1 \leq i < j \leq 30$.
\end{enumerate}
This procedure is repeated over $30$ text embeddings derived from prompts randomly sampled from the prompt dataset. The resulting empirical distribution of the distance ratios has mean 4.27 and standard deviation 0.44, and is illustrated in Figure \ref{fig:ratio_distr}. This empirically validates our claim in Lemma \ref{lemma:random_latent_rs} that pairwise distances between generated latents are lower bounded by a constant factor of the distances between their initial starting points.
\begin{figure}[!htbp]
    \centering
    \includegraphics[width=1\linewidth]{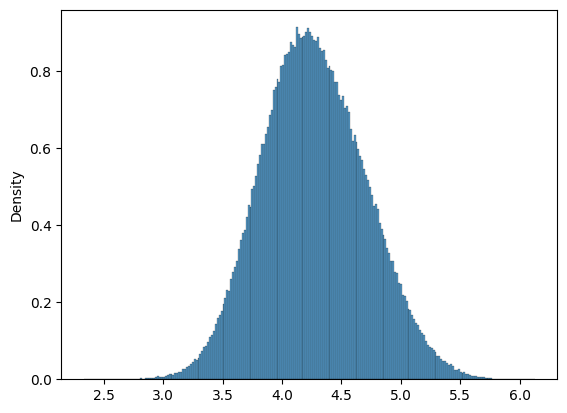}
    \caption{Distribution of the ratio between pairwise $\ell_2$ distances of generated latents and the corresponding $\ell_2$ distances between their starting points.}
    \label{fig:ratio_distr}
\end{figure}

\section{Additional Evaluation} \label{appendix:additional_eval}

\subsection{Distribution of Similarity Function} \label{section:expt_sim}

Assumption A2 posits that given a generated object $I=L_m(e, s_0)$, the distribution of $Sim(I, L_m(e, s))$, where $s \sim \mathcal{D}_0$, follows a sub-exponential distribution. Since the sub-exponential distribution is intractable, we assumed in Section \ref{subsection:similarity} that a generalized normal distribution closely approximates this sub-exponential distribution. To validate these two assumptions, we perform the following: Given a fixed prompt embedding $e$ and model meta-parameters $m$,
\begin{enumerate}
    \item Independently sample $s_0, s_1, \dots, s_{n} \sim \mathcal{N}(0, 1)$, where $n = \ln^{2} \frac{1}{\alpha} \ln \frac{1}{\delta}$, $\alpha = 2^{-10}$ and $\delta = 10^{-3}$
    \item Generate $I_0, \cdots, I_{n} $, where $I_j = L_m(e, s_j)$ denotes the (flattened) latent generated conditional on $m, e$ and starting point $s_j$.
    \item Compute the similarity scores $T = Sim(I_0, I_j) = \frac{1}{d} I_0 \cdot I_J$ for $j=1, \dots, {n} $, where $d = 4\times 64 \times 64$. 
    \item Approximate the scale, location and shape parameters of the generalized normal distribution.
    \item Compute the Kolmogorov-Smirnov (KS) distance between the computed similarity scores and the approximated distribution
\end{enumerate}
This procedure is repeated over $100$ text embeddings generated from randomly sampled prompts. 

Given two cumulative distribution functions (CDF) $F$ and $\hat{F}$ over support $\mathbb{X}$, the KS distance between $F$ and $\hat{F}$ is:
$
    KS(F, \hat{F}) = \sup_{x \in \mathbb{X}} |F(x) - \hat{F}(x)|.
$
Intuitively, $KS(F, \hat{F})$ captures the maximum discrepancy between $F$ and $\hat{F}$, small values indicate that the two distributions are close. The mean and standard deviation of the KS distances across the 100 embeddings are 0.012 and 0.003, respectively, and the distribution is illustrated in Figure \ref{fig:ks}. This empirically confirms that the sub-exponential distribution closely approximates the distribution of the similarity scores. 

\begin{figure}
     \centering
     \includegraphics[width=.85\linewidth]{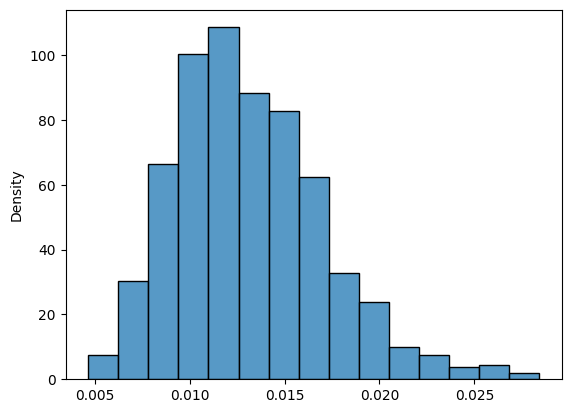}
     \caption{Empirical distribution of KS distances, which observes that the fitted sub-exponential distribution closely approximates the distribution of the similarity scores.}
     \label{fig:ks}
 \end{figure}

\section{Impact of POA Scheme on Generation Process} \label{appendix:samples_generated} 
In our empirical evaluations, we only modified the $\mathcal{N}(0, I)$ sampling function of the LDM, leaving all other components unchanged. Therefore, we expect this scheme to work seamlessly for all LDMs which samples a random starting point from this distribution. 

Existing watermarking methods, such as Gaussian shading and tree-ring, have been shown to introduce detectable artifacts and can be easily removed. In contrast, since the $\mathcal{N}(0, I)$ sampling function relies on a PRG, we do not face these issues. 

Recall that in the POA framework, the random starting point is generated by passing $f_i(\langle m, e,r\rangle)$ into the PRG. Diversity in the generated outputs can still be ensured by varying the free bits $r$, even when the prompt and embedding remain fixed. By the PRF assumption, as $r$ varies, the values $f_i(\langle m, e,r\rangle)$ are computationally indistinguishable from truly random strings. Consequently, the resulting starting points generated by the POA scheme are indistinguishable from those obtained using truly random seeds.

Figure \ref{fig:generated_images} illustrates examples of images generated with identical prompts and meta-parameters but different values of $r$.

\begin{figure*}[!htbp]
\begin{subfigure}[b]{\linewidth}
\centering
    \includegraphics[width=0.165\linewidth]{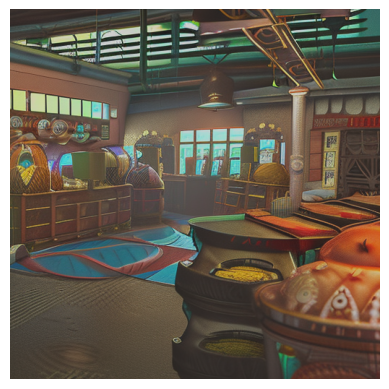}
    \includegraphics[width=0.165\linewidth]{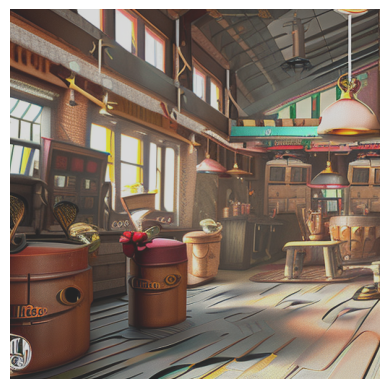}
    \includegraphics[width=0.165\linewidth]{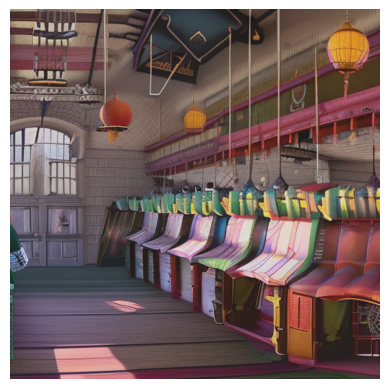}
    \includegraphics[width=0.165\linewidth]{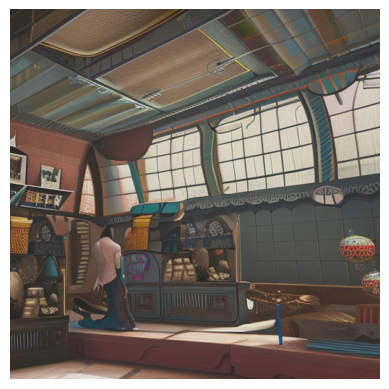}
    \includegraphics[width=0.165\linewidth]{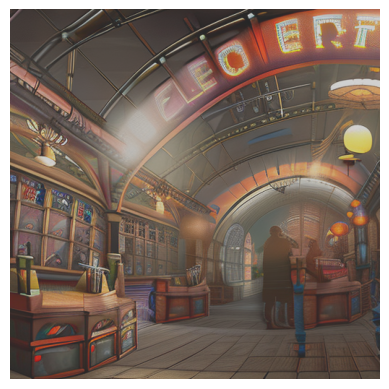}
    \\
    \includegraphics[width=0.165\linewidth]{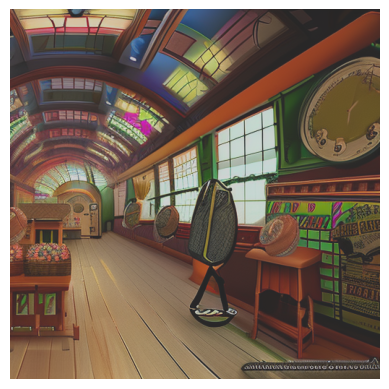}
    \includegraphics[width=0.165\linewidth]{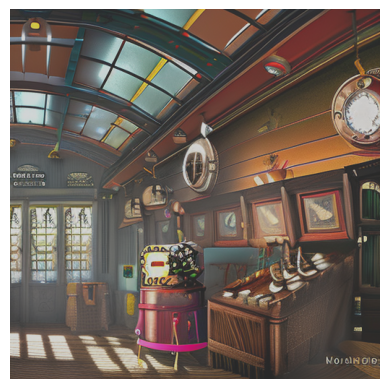}
    \includegraphics[width=0.165\linewidth]{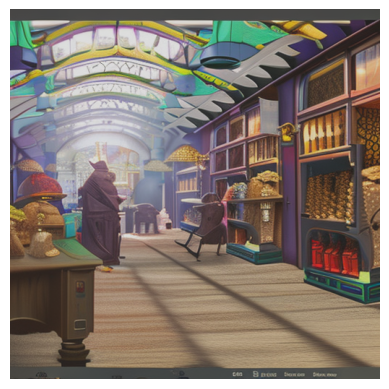}
    \includegraphics[width=0.165\linewidth]{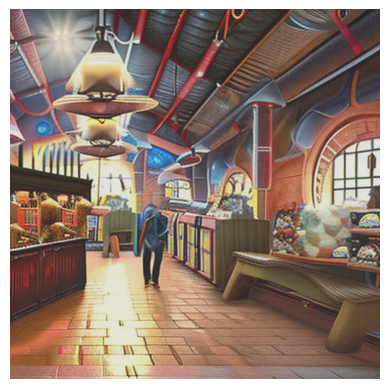}
    \includegraphics[width=0.165\linewidth]{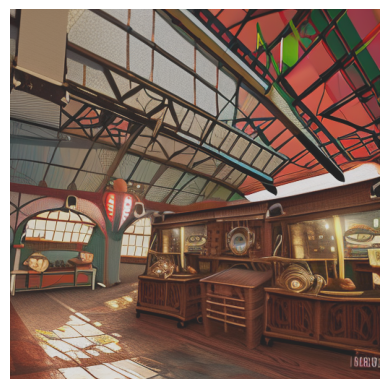}
    \caption{``steampunk market interior, colorful, 3 d scene, greg rutkowski, zabrocki, karlkka, jayison devadas, trending on artstation, 8 k, ultra wide angle, zenith view, pincushion lens effect''}
\end{subfigure}

\begin{subfigure}[b]{\linewidth}
\centering

    \includegraphics[width=0.165\linewidth]{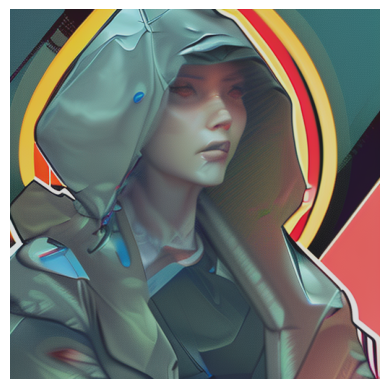}
    \includegraphics[width=0.165\linewidth]{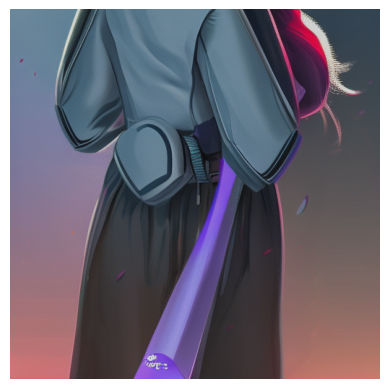}
    \includegraphics[width=0.165\linewidth]{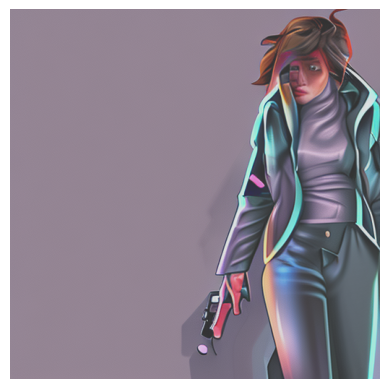}
    \includegraphics[width=0.165\linewidth]{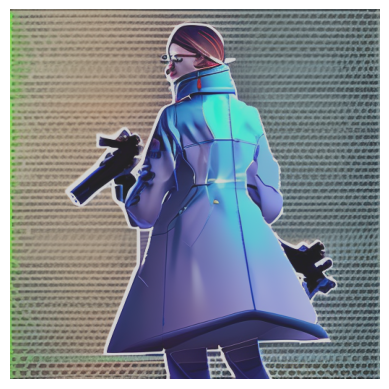}
    \includegraphics[width=0.165\linewidth]{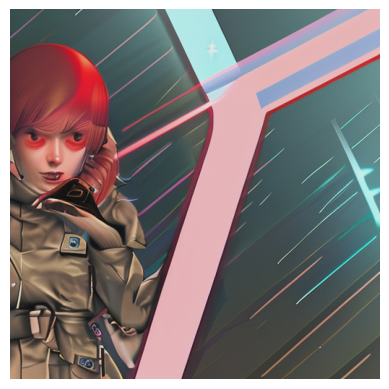}
    \\
    \includegraphics[width=0.165\linewidth]{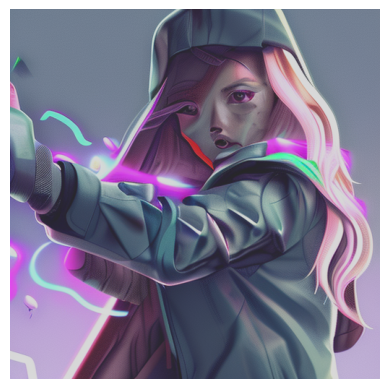}
    \includegraphics[width=0.165\linewidth]{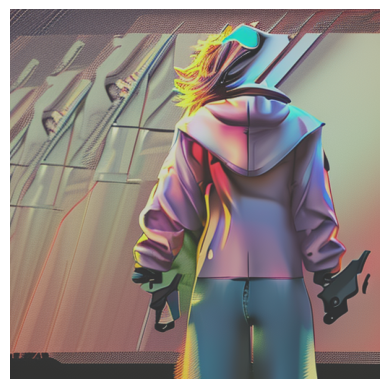}
    \includegraphics[width=0.165\linewidth]{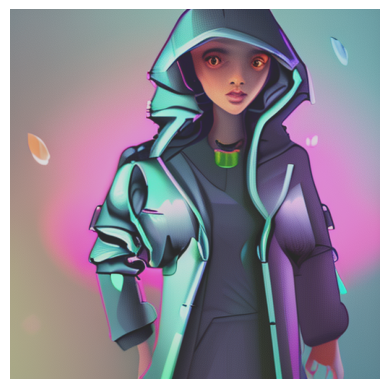}
    \includegraphics[width=0.165\linewidth]{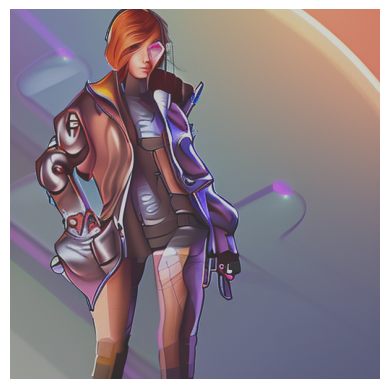}
    \includegraphics[width=0.165\linewidth]{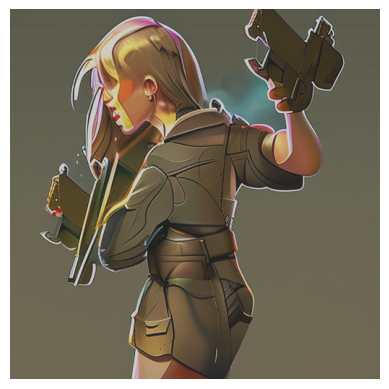}
    \caption{``girl in a futuristic raincoat, holding a revolver, character concept art, valorant game style, digital art, many details, super realistic, high quality, 8 k''}
\end{subfigure}

\begin{subfigure}[b]{\linewidth}
\centering

    \includegraphics[width=0.165\linewidth]{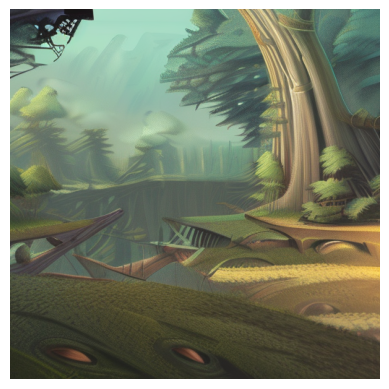}
    \includegraphics[width=0.165\linewidth]{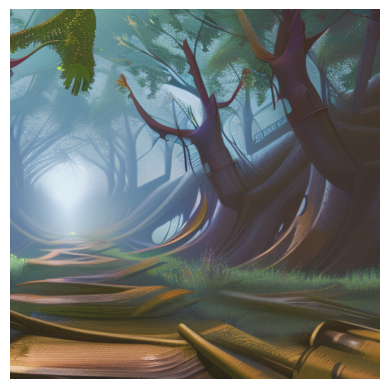}
    \includegraphics[width=0.165\linewidth]{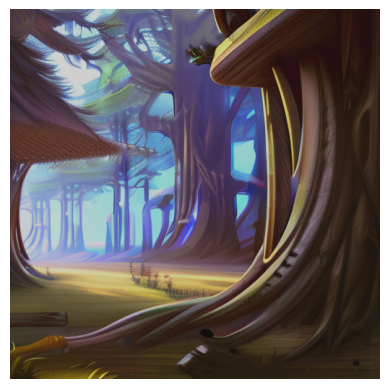}
    \includegraphics[width=0.165\linewidth]{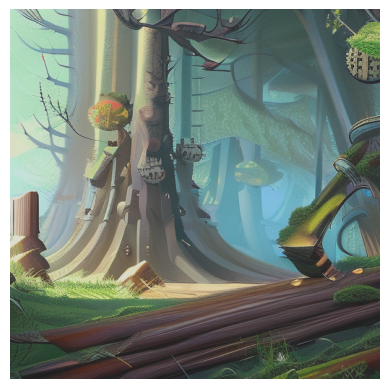}
    \includegraphics[width=0.165\linewidth]{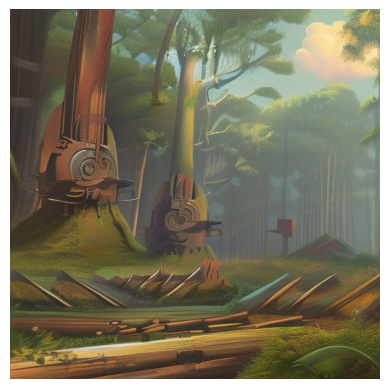}
    \\
    \includegraphics[width=0.165\linewidth]{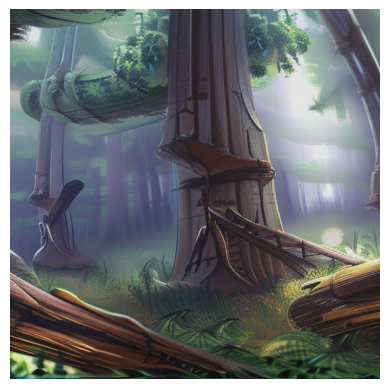}
    \includegraphics[width=0.165\linewidth]{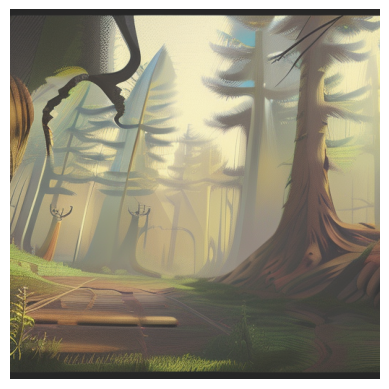}
    \includegraphics[width=0.165\linewidth]{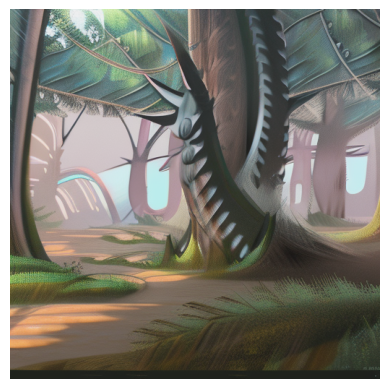}
    \includegraphics[width=0.165\linewidth]{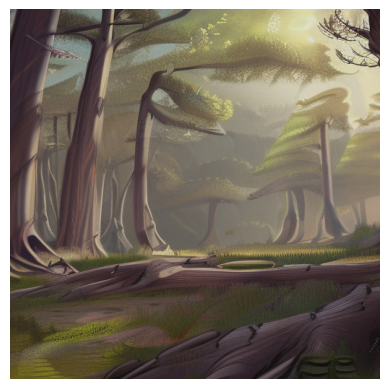}
    \includegraphics[width=0.165\linewidth]{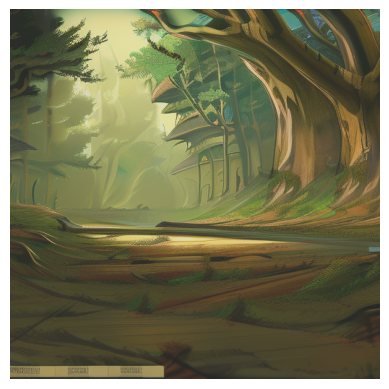}
    \caption{``good day, morning, landscape, no people, no man, steampunk, forest, wood, lost world, Anime Background, concept art, illustration,smooth, sharp focus, intricate, super wide angle, trending on artstation, trending on deviantart, 4K''}
\end{subfigure}
    \caption{Examples of images generated using the same and meta-parameters $m$ and prompt embeddings $e$ (the individual prompts are given in the captions)  but different free-bits $ r $. Despite fixed conditioning, variation in $ r $ yields diverse outputs as each different $r$ corresponds to a diferent starting point independently sampled from $\mathcal{N}(0,I)$. }

    \label{fig:generated_images}
\end{figure*}

% \section{Summary of Changes Made for Minor Revision}
% \subfile{changes_appendix.tex}  

\end{document}